\documentclass{article}[11pt]

\usepackage[fencedCode,inlineFootnotes,citations,definitionLists,hashEnumerators,smartEllipses,hybrid]{markdown}
\usepackage[utf8]{inputenc}
\usepackage[T1]{fontenc}
\usepackage[scr=txupr]{mathalfa}
\usepackage[british]{babel}
\usepackage[autostyle]{csquotes}
\usepackage{mathtools}
\usepackage{amsmath, amsthm,amssymb,amstext,amsbsy} 
\usepackage{tabularx}
\usepackage{thmtools,thm-restate}
\usepackage{fullpage}
\usepackage{todonotes}[color=white]
\usepackage{enumitem}
\usepackage{caption}
\setlist[description]{leftmargin=15pt,labelindent=15pt}

\usepackage[bookmarks=false]{hyperref}  
\definecolor{darkgreen}{rgb}{0,0.5,0}
\hypersetup{linktocpage=true,colorlinks=true,citecolor=red,linkcolor=darkgreen}

\makeatletter
\renewcommand\paragraph{\@startsection{paragraph}{4}{\z@}%
                                    {1.5ex \@plus1ex \@minus.2ex}%
                                    {-1em}%
                                    {\normalfont\normalsize\bfseries}}

\makeatother
\usepackage{setspace}

\newif\ifnotes\notestrue

\ifnotes
 \usepackage{color}
 \definecolor{mygrey}{gray}{0.50}
 \newcommand{\notename}[2]{{{\footnotesize{\bf (#1:} {#2}{\bf ) }}}}
 \newcommand{\noteswarning}{{\begin{center} {\Large WARNING: NOTES ON}\end{center}}}
 
\else
 
 \newcommand{\notename}[2]{{}}
 \newcommand{\noteswarning}{{}}
\fi

\newtheorem{theorem}{Theorem}[section]
\newtheorem{claim}[theorem]{Claim}

\newtheorem{lemma}[theorem]{Lemma}

\newtheorem{remark}[theorem]{Remark}

\newtheorem{thm}{Theorem}[section]

\newcommand{\expref}[2]{{\texorpdfstring{\hyperref[#2]{#1~\ref{#2}}}{#1~\ref{#2}}}} 
\newcommand{\secref}[1]{\expref{Section}{#1}}
\newcommand{\thmref}[1]{\expref{Theorem}{#1}}
\newcommand{\clmref}[1]{\expref{Claim}{#1}}
\newcommand{\tref}[1]{\expref{Theorem}{#1}}
\newcommand{\appref}[1]{\expref{Appendix}{#1}}
\newcommand{\lref}[1]{\expref{Lemma}{#1}}
\newcommand{\corref}[1]{\expref{Corollary}{#1}}
\newcommand{\conjref}[1]{\expref{Conjecture}{#1}}
\newcommand{\pref}[1]{\expref{Proposition}{#1}}
\newcommand{\figref}[1]{\expref{Figure}{#1}}

\newcommand{\boldp}{\ensuremath{\mathsf{p}}}

\newcommand{\ch}{\mathcal{C}}
\newcommand\defeq{\stackrel{\mathclap{\small\mbox{def}}}{=}}
\newcommand{\RZ}{\mathbb{Z}}
\newcommand{\R}{\mathbb{R}}
\newcommand{\E}{\mathbb{E}}
\newcommand{\N}{\mathbb{N}}
\newcommand{\FS}{\mathfrak{S}}
\newcommand{\CV}{\mathcal{V}}
\newcommand{\CD}{\mathcal{D}}
\newcommand{\ip}[2]{\langle #1,#2\rangle}
\newcommand{\sN}{\mathsf{N}}
\newcommand{\BR}{\mathbb{R}}
\newcommand{\BE}{\mathbb{E}}
\newcommand{\BP}{\mathbb{P}}
\newcommand{\CP}{\mathfrak{P}}
\newcommand{\bits}{\{0,1\}}
\newcommand{\pmone}{\{\pm1\}}
\newcommand{\ind}{\mathbf{1}}
\newcommand{\RN}{\mathbb{N}}
\newcommand{\Psijk}[1]{\Psi_{\mathbf{#1}}}
\newcommand{\CH}{\mathcal{H}}
\newcommand{\calP}{\ensuremath{\mathcal{P}}}
\newcommand{\CI}{\mathcal{I}}
\newcommand{\CE}{\mathcal{E}}
\newcommand{\bx}{\mathbf{x}}
\newcommand*\bh{\ensuremath{\boldsymbol{h}}}
\newcommand*\bj{\ensuremath{\boldsymbol{j}}}
\newcommand*\bk{\ensuremath{\boldsymbol{k}}}
\newcommand*\bl{\ensuremath{\boldsymbol\ell}}
\newcommand*\bm{\ensuremath{\boldsymbol{m}}}
\newcommand{\fhat}{\widehat{f}}
\newcommand{\dhat}{\widehat{d}}
\newcommand{\polylog}{\mathrm{polylog}}
\newcommand{\poly}{\mathrm{poly}}
\newcommand{\infnorm}[1]{\left\| #1 \right\|_{\infty}}
\newcommand{\eps}{\epsilon}
\newcommand{\disc}{\mathsf{disc}}
\newcommand{\sign}{\mathsf{sign}}
\newcommand{\CB}{\mathcal{B}}
\newcommand{\CL}{\mathcal{L}}
\newcommand{\CS}{\mathcal{S}}
\newcommand{\CN}{\mathcal{N}}
\newcommand{\CZ}{\mathcal{Z}}
\newcommand{\sfp}{\mathsf{p}}
\newcommand{\comp}[1]{\overline{#1}}
\newcommand{\err}{\mathsf{err}}
\newcommand{\CG}{\mathcal G}
\newcommand{\CT}{T}
\newcommand{\calT}{\mathcal{T}}
\newcommand{\smin}{\eps_\mathsf{min}}
\newcommand{\smax}{\eps_\mathsf{max}}
\newcommand{\cov}{\mathbf{\Sigma}}
\newcommand{\error}{\mathsf{err}}
\global\long\def\norm#1{\left\Vert #1\right\Vert }
\newcommand{\leaf}{l}
\newcommand{\bK}{\overline{K}}
\newcommand{\BS}{\mathbb{S}}
\newcommand{\p}{\mathbb{P}}
\newcommand{\sfq}{\mathsf{q}}
\newcommand{\sfu}{\mathsf{u}}
\newcommand{\tM}{M^+}
\newcommand{\op}{\mathsf{op}}
\newcommand{\diam}{\mathsf{diam}}
\newcommand{\im}{\mathrm{im}}
\newcommand{\SE}{\mathscr{S}}
\newcommand{\SA}{\mathscr{A}}
\newcommand{\SB}{\mathscr{B}}
\newcommand{\SG}{\mathscr{G}}
\newcommand{\ST}{\mathscr{T}}
\newcommand{\Tr}{\mathrm{Tr}}
\newcommand{\SZ}{\mathscr{Z}}
\newcommand{\SD}{\mathscr{D}}
\newcommand{\SL}{\mathscr{L}}
\newcommand{\arbnorm}[1]{\left\|#1\right\|_*}
\newcommand{\qlt}{q^l_{j,k}}
\newcommand{\qrt}{q^r_{j,k}}
\newcommand{\dlt}{d^l_{j,k}}
\newcommand{\drt}{d^r_{j,k}}
\newcommand{\SM}{\mathscr{M}}
\newcommand{\SN}{\mathscr{N}}
\newcommand{\ovl}{\overline{L}}
\newcommand{\putat}[3]{\begin{picture}(0,0)(0,0)\put(#1,#2){#3}\end{picture}}

\DeclareMathOperator{\nl}{nl}
\DeclareMathOperator{\bin}{bin}
\DeclareMathOperator{\sgn}{sgn}
\DeclareMathOperator{\total}{Total}
\DeclareMathOperator{\child}{child}
\DeclareMathOperator{\sib}{sib}
\DeclareMathOperator{\gap}{gap}
\DeclareMathOperator{\conge}{cong}
\DeclareMathOperator{\dist}{dist}

\let\oldabstract\abstract
\let\oldendabstract\endabstract
\makeatletter
\renewenvironment{abstract}
{%
               {\list{}{\addtolength{\leftmargin}{1em} 
                        \listparindent 1.5em%
                        \itemindent    \listparindent%
                        \rightmargin   \leftmargin%
                        \parsep        \z@ \@plus\p@}%
                \item\relax}%
               {\endlist}%
\oldabstract}
{\oldendabstract}
\makeatother

\allowdisplaybreaks

\title{The Power of Two Choices in Graphical Allocation}
\author{Nikhil Bansal\thanks{CWI Amsterdam and TU Eindhoven, \texttt{N.Bansal@cwi.nl}. Supported by the  NWO VICI grant 639.023.812.} \and {Ohad Feldheim\thanks{Hebrew University of Jerusalem Israel, \texttt{ohad.feldheim@mail.huji.ac.il}. Supported by ISF grant 1327/19.}}
}

\date{}

\begin{document}

\maketitle

\begin{abstract}
\medskip
 The graphical balls-into-bins process is a generalization of the classical 2-choice balls-into-bins process, 
 where the bins correspond to vertices of an arbitrary underlying graph $G$. At each time step an edge of $G$ is chosen uniformly at random, and a ball must be assigned to either of the two endpoints of this edge.
The standard 2-choice process corresponds to the case of $G=K_n$.
 
 For any $k(n)$-edge-connected, $d(n)$-regular graph on $n$ vertices, and any number of balls,
 we give an allocation strategy 
that, with high probability, ensures a gap of $O((d/k) \log^4\hspace{-1pt}n \log \log n)$, between the load of any two bins. 
In particular, this implies polylogarithmic bounds for natural graphs such as cycles and tori, for which the classical greedy allocation strategy is conjectured to have a polynomial gap between the bins' loads. 
For every graph $G$, we also show an $\Omega((d/k) + \log n)$ lower bound on the gap achievable by any allocation strategy. This implies that our strategy achieves the optimal gap, up to polylogarithmic factors,  for every graph $G$.

Our allocation algorithm is simple to implement and requires only $O(\log(n))$ time per allocation. It can be viewed as a more global version of the greedy strategy that compares average load on certain fixed sets of vertices, rather than on individual vertices.
A key idea is to relate the problem of designing a good allocation strategy to that of finding suitable multi-commodity flows. To this end, we  consider R\"{a}cke's cut-based decomposition tree and define certain orthogonal flows on it.

    
\end{abstract}
\thispagestyle{empty}
\clearpage
\newpage
\setcounter{page}{1}

\section{Introduction} 

Randomized balls-into-bins models serve as useful abstractions for various problems arising in hashing, load balancing and resource allocation in parallel and distributed systems and have been extensively studied in the areas of probability, economics and algorithms (see e.g.,~\cite{RS98,DR96,AK14}).
The balls typically represent tasks or items, that need to be allocated to resources that are modeled by the bins. The goal is to balance the loads across bins as much as possible.

These models usually differ by the kind of control available to the algorithm over the allocation process.
In the classical {\em single-choice} model, the algorithm has no control and each ball is placed in a bin chosen uniformly at random. For $n$ balls and $n$ bins, it is well known that the heaviest bin has load $(1+o(1)) \ln n/\ln \ln n$ with high probability (w.h.p.). In the {\em heavily loaded case}, where the number of balls $T$ can be much larger than $n$, 
the bin loads are in the range $T/n \pm \Theta(\sqrt{(T \log n)/n})$  (provided  that $T \geq n \log n$). In particular, the deviation from the average load of $T/n$ increases with $T$ as $\Theta(\sqrt{(T\log n)/n})$. 

\paragraph{2-choice model.} Perhaps the simplest and most well-studied controlled variant of this model is the\linebreak \emph{2-choice} model. Here at each step the algorithm is given two uniformly chosen bins into one of which it must allocate the ball. This seemingly minor modification, leads to 
substantial improvements. In a seminal result, Azar, Broder, Karlin and Upfal \cite{ABKU94} showed that if a ball is placed in the least loaded of $d\geq 2$ uniformly sampled bins, then for $T=O(n)$, the maximum load reduces to $\ln \ln n/\ln d + \Theta(1)$. They also establish that this {\em greedy} allocation strategy is asymptotically optimal for this model.

These results were extended by Berenbrink, Czumaj, Steger and V{\"o}cking~\cite{BCSV06} to the 
harder setting of
arbitrary $T$. They showed that the maximum load is
$T/n + \ln \ln n/\ln d + O(1)$ with probability\linebreak $1-1/\text{poly}(n)$. Remarkably this implies that already for $d=2$ choices, the excess load over the average does not increase with $T$ at all (unlike for the case of $d=1$). A simpler proof of this result, albeit with worse tail bounds, was given by Talwar and Wieder [TW14]).

\paragraph{Graphical process.}
In many natural settings, there are restrictions on which pairs of bins can be queried or where the ball can be placed. An elegant generalization of the $2$-choice process, called the {\em graphical process}, was introduced by Kenthapadi and Panigrahy \cite{KP06}.
Here there is an underlying graph $G=(V,E)$ on $n=|V|$ vertices, and 
at each step, a uniformly random edge $e=(u,v)$ is chosen and the ball must be placed on one of the two endpoints of $e$. Notice that the standard $2$-choice process corresponds to the special case of $G=K_n$.
This motivates the following natural question:

{\em Given a graph $G$ what is the best load balance obtainable by a graphical two-choice allocation strategy?}

Extending the results for the classical $2$-choice process ($G=K_n$), 
\cite{KP06} showed that if $G$ is $n^{\epsilon}$-regular, then for $T=n$ balls, the greedy strategy, has  maximum load is $\ln \ln n + O(\log 1/\eps)$. An extension to hypergraphs was considered in \cite{Godfrey08}. 

The harder setting of arbitrary $T$, which is also the focus of our work, was considered by Peres, Talwar, Weider \cite{PTW15}. They investigated the greedy strategy,
and using an elegant majorization and potential function technique showed that the gap  between the maximum and minimum bin load is $\Theta(\log n)$  for complete graphs and
more generally, $O((\log n)/\beta)$ 
for $d$-regular graphs with edge-expansion\footnote{For any subset $S \subset V$ with $|S|\leq n/2$, $E(S,\overline{S}) \geq \beta d |S|$.} $\beta$.

\paragraph{Upper gap versus gap.}
Notice that the objective studied by \cite{PTW15} is slightly different from the one
in \cite{ABKU94,BCSV06,KP06}. Let us use the term {\em upper gap} for the difference between the maximum load and the average load, and {\em gap} for the difference between the maximum and minimum load. 
These objectives can differ sometimes, e.g.,~for $G=K_n$, the gap\footnote{The lower bound of $\Omega(\log n)$ on the gap follows (in fact for any graph $G$) because by a standard coupon collector argument, in any consecutive sequence of $\Omega(n\log n)$ steps, with constant probability some bin will not be considered at all.} is $\Theta(\log n)$ while the upper gap is $\Theta(\log \log n)$. 
However, this difference seems less relevant for general graphs, 
e.g.,~for a constant degree expander $G$, the gap is $\Theta(\log n)$ while the upper gap is $\Omega(\log n/\log \log n)$. We discuss this issue further in Sections \ref{sec:related} and \ref{sec:lb}.

\paragraph{Limitations of the greedy strategy.} While almost all previous results on the 2-choice model consider
the greedy strategy, 
it turns out that greedy can be quite sub-optimal 
for general graphs.
For example, for a $n$-vertex cycle $G$ the conjectured gap and upper gap under greedy are $\Theta(\sqrt{n})$ \cite{ANS20}\footnote{The authors consider a model where after allocating the ball to a vertices of the requested edge, the load on the two vertices is averaged. In this model, which intuitively should have a lower gap than greedy, they show that the typical gap is $\Omega(\sqrt{n})$.}. See also our simulation results in Figure~\ref{fig:Greedy Gap} below.
The best known upper bound for a cycle is $O(n \log n)$ by the result of \cite{PTW15} (as the expansion $\beta = O(1/n)$). 


Analyzing the gap for greedy, even on simple graphs like cycles, seems surprisingly hard and its study has led to several intriguing conjectures in asymmetric stochastic dynamics that seem beyond current techniques.
On the cycle, the loads are conjectured to converge after proper scaling to a Brownian motion.
For more general graphs 
the load fluctuations under greedy are conjectured to be similar to the fluctuations in classical statistical mechanics models.\footnote{In particular Peres (in private communication) suggested that, up to a $\log n$ factor the gap should be the same as that of a Gaussian free field on $G$. This suggests that even for certain non-expander graphs, such as high-dimensional balanced grids, the greedy algorithm should have polylogarithmic gaps.}
\begin{figure}[h!]
    \centering
    \includegraphics[scale=0.5]{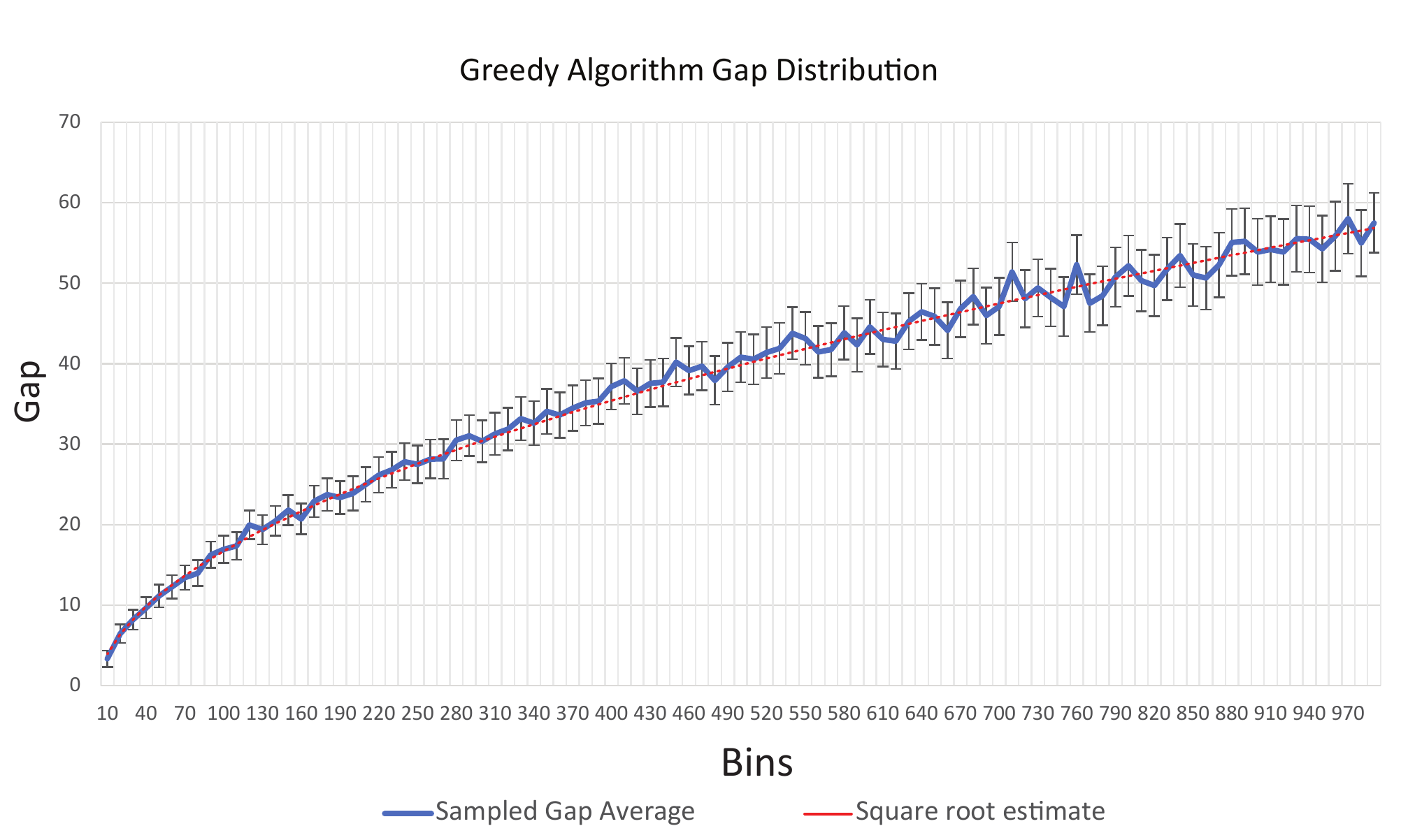}
    \caption{The gap for greedy algorithm on cycles of sizes 10 to 1000, averaged over 84 runs of $10^9$ balls. The dotted line is the function $f(x)=1.85\sqrt{x}-1$, error margins for 95\% confidence are provided. The graph clearly shows the polynomial growth of the gap.
    }
    \label{fig:Greedy Gap}
\end{figure}
 
\textbf{Beyond greedy strategy.}
An advantage of the greedy strategy is that it is only uses local information: all we need to know in order to allocate is the load of the two suggested bins. However, in many applications 
the algorithm may be able to gather more  global information to make better decisions, even though the task (ball) itself has only a few available options (bins). 
For example, in some situations transferring the task may require heavy overhead, complicated real world shipment, or overcoming physical restrictions, while the information could propagate in the network at much lower cost.
Hence, a natural question is whether there are other allocation strategies which are superior and achieve substantially better gaps than greedy on more general graphs. Ideally we would also like that these strategies be simple to implement and be realizable in a distributed setting with low communication overhead.

\subsection{Results}
\label{sec:results}

Our main result is an allocation strategy that achieves the best possible gap, up to polylogarithmic factors, for the graphical process on any graph $G$. 
More formally, we show the following. 
\begin{thm}\label{thm:main}
Let $G=(V,E)$ be any $k$ edge-connected, $d$-regular\footnote{The regularity assumption on $G$ is standard and ensures uniform expected load on all bins under a random strategy. The results extend to irregular graphs in a natural way where one measures the gap with loads normalized suitably by the degree.
} graph on $n$ vertices. There is an allocation strategy for the graphical process on $G$, that guarantees for any time $t\in\N$,  \[\gap_G(t) = O((d/k) \log^4 n \log \log n)\]
with probability at least $1-1/\text{poly}(n)$. Here $\gap_G(t)$ is the maximum difference in vertex loads at time $t$. 
\end{thm}
This is also the best possible 
gap for every graph $G$ up to $\text{polylog}(n)$ factors.
\begin{thm}
\label{thm:lb}
For any $d$-regular graph $G$ and for any allocation strategy for $G$, for any time $t$, with constant probability, $gap_G(t) = \Omega(d/k + \log n)$, where $k$ is the edge-connectivity of $G$.
\end{thm}

In particular Theorem \ref{thm:main} gives an allocation strategy with gap $\text{polylog}(n)$ for cycles and grids and any arbitrary connected bounded degree graph.
More  generally, it implies polylogarithmic gap for any  graph $G$ with connectivity $k = \Omega(d/\text{polylog}(n))$, without any requirements on the expansion of $G$.
 
For a graph $G$ the bound the gap in Theorem \ref{thm:main} is actually $O(\alpha_G\, (d/k) \log^2 n )$ where $\alpha_G$ is the congestion ratio for oblivious routing on $G$ based on a R\"{a}cke decomposition tree \cite{R02}\footnote{While better bounds are known for congestion in oblivious routing  \cite{AzarCFKR04, Racke08, AndersenFeige} based on a convex combination of trees, our method requires a single 
R\"acke tree as the demands for our flows depend on the tree itself.}.
It is known that \linebreak $\alpha_G = O(\log^2 n \log \log n)$ for general graphs $G$ \cite{HHR03} and better bounds are also known for specific graphs, e.g.,~$\alpha_G=O(\log n)$ for a cycle.

\paragraph{The Algorithm.}
The algorithm consists of two parts: a preprocessing stage and an online allocation stage. 
 The allocation strategy stays fixed throughout, is quite simple to implement and incurs only $O(\log n)$ worst case running time per allocation. It
 can be viewed as a more {\em global} version of the greedy strategy, where given an edge $e=(u,v)$, instead of comparing the loads on $u$ and $v$, it compares the average load on two sets selected at random from a fixed small collection and allocates the ball to $u$ or $v$ according to the result.

The preprocessing stage computes (i) a binary hierarchical decomposition of the vertices of $G$ and (ii)
for each edge $e$ a distribution $P_e$ over pairs of sibling sets $(S,S')$ in the decomposition.
For each $e$, the support of $P_e$ is only $O(\log n)$.
The sets $S$ and the distributions $P_e$ 
are then fixed and do not change over time.

In the allocation stage,
upon the request for random edge $e=(u,v)$,  a pair of sets $(S,S')$ is sampled according to $P_e$ and the ball is assigned to $u$ or $v$ depending on whether $S$ or $S'$ has larger average load. 

\paragraph{Complexity.} The allocation strategy is very efficient to implement and requires only $O(\log n)$ time per allocation step. 
In fact, this can be further reduced to $O(1)$ amortized updates per allocation step.
The total space used is $O(m \log n)$. 

The allocation strategy is also quite robust and does not require the exact knowledge of load on the sets $S$ in the decomposition.
We use this to give a {\em distributed} implementation the allocation strategy that requires only $O(1)$ amortized messages per allocation, with $O(\log n)$ bits per message.

These implementation details are discussed in Section \ref{s:implement}.


\begin{remark} Finally we remark that our results do not require the full power of two choices. In the $1+\beta$ graphical choice model (see also Section~\ref{sec:related} below), at each step an edge is given only with probability $\beta$, and with probability $1-\beta$ there is no choice and the ball is allocated to a random bin. Our bound on the gap extends directly to this model with factor $O(1/\beta)$ loss.
\end{remark}

\subsection{Preliminaries, overview and techniques}\label{sec:overview}
We now give a detailed overview of the ideas and the algorithm. We begin by describing the relevant notation.

\paragraph{Notation.}
A graphical process is specified by some fixed $d$-regular graph $G=(V,E)$. Henceforth we denote $n=|V|$ and $m=|E|. $\footnote{typically in balls-into-bins literature, $m$ is used for the number of balls, but we will think of the allocation process as indefinite, using $t$ for the index of the allocated ball.} 
Let $N(u)$ denote the neighborhood of the vertex $u$.
At each time $t=1,2,\ldots,$ an  edge $e=e_t=(u,v) \in E$ is chosen (\emph{requested}) uniformly at random, and a ball must be assigned to one of its endpoints: $u$ or $v$. We refer to the vertices of $G$ as bins, and say that vertex $u$ has load $\ell$ at time $t$, if $\ell$ balls have been assigned to it after $t$ steps.
Let $L^t:V \rightarrow \N$ denote the load vector at time $t$. 
We assume $L^0(u)=0$ for all $u \in V$ so that the total load $\|L^t\|_1 =t$ for each $t\in \N$ and the average vertex load is $t/n$. 

An allocation strategy, upon the request $e=(u,v)$, decides whether to assign the ball to $u$ or $v$ based on the current loads (and possibly on the entire history so far). The goal of the strategy is to regulate the process in order to keep  the load gap small, where the {\em gap} at time $t$ is defined as 
\[ \gap(t)=\gap_G(t) = \max_u L^t(u)  - \min_u L^t(u).\]
As $ \max_u | L^t(u) - t/n| \leq \gap_G(t) \leq  2 \max_u |L^t(u)-t/n|$, we sometimes work with $\max_u | L^t(u) - t/n|$, the maximum deviation from the average load.
For a subset $S \subset V$, 
We use $L^t(S) := \sum_{u \in S} L^t(u)$ and $\overline{L}^t(S) := L^t(S)/|S|$ to
 the total and average load on vertices in $S$ at time $t$, respectively. 


\paragraph{Edge biases, induced drift and flows.} 
Upon the arrival of an edge $e=(u,v)$, by choosing whether to allocate the ball to $u$ or $v$, an allocation strategy can {\em bias} the load toward $u$ or $v$.
Let $b_t(u,v)$ denote the bias towards $v$ for an edge $e=(u,v)$ at time $t$. So $b_t(u,v)= 2p-1 \in [-1,1]$ if the ball is allocated to $v$ with probability $p$ upon arrival of $e$. Thus, $b_t(u,v)=-b_t(v,u)$.

At any time $t$, the behavior of an allocation strategy defines the biases $b_t(u,v)\in[-1,1]$ for each edge $(u,v) \in E$, and, conversely, any valid biases at time $t$ define an allocation strategy.
Now let us consider the effect of the biases at a vertex. At each time step, the average load over all bins increases by $1/n$.
 Let \[q_t(v) =  \frac{1}{m} \sum_{u \in N(v)} \frac{(1+ b_t(u,v))}{2}  =  \frac{1}{n} + \sum_{u \in N(v)} \frac{b_t(u,v)}{2m}\] 
 denote the probability of assigning a ball to $v$ at time $t$.
 Let us call $d_t(v)=2m( q_t(v)-1/n) = \sum_{u\in N(v) } b_t(u,v)$ the \emph{drift} induced by the strategy at $v$ at time $t$. 
 Observe that $d_t(v) \in [-d,d]$ and $q_t(v) \in [0,2/n]$.
 
 A simple but useful observation is that if we view the biases $b_t(u,v)$ as a flow of $b_t(u,v)$ from $u$ to $v$, then the drifts correspond to the  total flow entering a vertex.
 Conversely, for any desired vertex drifts if there is a corresponding feasible flow with capacity at most $1$ per edge, viewing the flow as edge biases gives an allocation strategy achieving this drift vector.
 


\paragraph{A first attempt.}
To control the load gap,
we would like to set the edge biases at each time $t$, such that the underloaded vertices are more likely to receive a ball (say $q_t(v) \geq (1+\delta)/n$) while the overloaded vertices are less likely ($q_t(v) \leq  (1-\delta)/n$) to receive a ball. 
By a simple Markov chain argument, this ensures an $O(\delta^{-1} \log n)$ gap w.h.p.~at any time.
This is the same as assigning drift $\geq \delta d$ (resp.~$\leq -\delta d$) to underloaded (resp.~overloaded) vertices.


However, such a drift is impossible  to realize as a feasible flow, unless $G$ has large expansion. In particular, consider a cycle where the the vertices of the left half of are overloaded and those of the right half are underloaded. As the total flow that can cross the cut separating the two halves is at most $2$, the drifts can only have absolute value $O(1/n)$ on average, resulting in an $\Omega(n \log n)$ load deviation.
Roughly, this is the reason why the expansion condition on $G$ is required in \cite{PTW15}.

Additionally, we would also like the allocation strategy to be simple to implement, and avoid solving a flow problem at each time step, depending on which vertices are overloaded or underloaded.


We now describe our method for geting around both of these issues. Roughly speaking, we will create drifts for sets of vertices instead of drifts for individual vertices.
\paragraph{Hierarchical decomposition and orthogonal drifts.}
Let $G=(V,E)$ be the graph underlying the process.
Consider some {\em binary} hierarchical decomposition of $G$, represented by a binary tree $T = (V_T,E_T)$, with internal nodes $i \in V_T$ corresponding to subsets $S_i \subset V$ and leaves in $V_T$ to singleton sets, one for each vertex of $G$. We will identify the leaves of $T$ with the vertices of $G$. 
We use {\em node} to refer to the internal vertices of $T$ and {\em vertex} to refer to the vertices of $G$. 

Let $r$ be the root of $T$. For a leaf $u$, consider the unique path $u=i_0, i_1, \ldots, i_h=r$ from $u$ to $r$ in $T$, so that $\{u\}=S_{i_0} \subset S_{i_1} \subset \ldots \subset  S_{i_h}= V$. 
As $\overline{L}^t(S_r) = \overline{L}^t(V)= t/n$ and $\overline{L}^t(S_u)=L^t(u)$, the load deviation for $u$ at time $t$ can be upper bounded by  \[|L^t(u)-t/n |= \left|\sum_{j=1}^h \big(\overline{L}^t(S_{i_{j-1}})- \overline{L}^t(S_{i_j})\big) \right|\leq    \sum_{j=1}^h \left|\overline{L}^t(S_{i_{j-1}})- \overline{L}^t(S_{i_j})\right|,\]
and hence to control the load deviation for each $u \in V$, up to an $h=O(\log n)$ factor, it suffices to control the gap $|\overline{L}^t(S_{i_{j-1}})- \overline{L}^t(S_{i_j})|$ for each parent-child node pair.

Fix a non-leaf node $i\in V_T$, and let $\ell(i), r(i)$ denote its left and right children. 
A simple computation shows that the parent-child gap can be bounded in terms of the \emph{sibling gap} $|\overline{L}(S_{\ell(i)}) -  \overline{L}(S_{r(i)})|$.
To control this sibling gap for each pair of siblings $\ell(i)$ and $r(i)$ in the tree, we would like to create a (dynamic) \emph{balancing drift} between $S_{\ell(i)}$ and $S_{r(i)}$, that  at any time is directed towards the sibling with lower load. 

To avoid these different drifts from interfering,
a key idea is to construct the drifts to be \emph{orthogonal}, which ensures that the drift associated with $S_{\ell(i)}$ and $S_{r(i)}$ yields zero net drift on every other set $S_j\in T$. To do so we set the balancing drift to be the Haar measure associated with node $i$ in the  
hierarchical Haar decomposition basis corresponding to $T$. Simply put, we divide the drift on each of $S_{\ell(i)},S_{r(i)}$ evenly among all vertices of the corresponding set. 
An illustration of a system of orthogonal drifts is provided in Figure \ref{fig:edge flow}.

\paragraph{Small relative bias suffices.} The balancing drifts associated with $S_{\ell(i)}$ and $S_{r(i)}$ should also be sufficiently strong in order to get good bounds on $|\overline{L}(S_{\ell(i)}) - \overline{L}(S_{r(i)})|$, but on the other hand, the resulting flow (edge biases) should be  feasible. 
Suppose $k=\Omega(d)$ for this discussion.
A second key observation is that to achieve the polylogarithmic gap in Theorem \ref{thm:main}, it suffices to have the {\em total} drift between every two siblings 
$S_{\ell(i)}$ and $S_{r(i)}$ be $\Omega(d/\text{polylog}(n))$, i.e., the drifts for vertices in $S_{\ell(i)}, S_{r(i)}$ need only be $\Omega(d/|S_{\ell(i)}|\text{polylog}(n))$ and $\Omega(d/|S_{r(i)}|\text{polylog}(n))$ on average.
It may seem counter-intuitive at first, that the total drift does not need to scale with the sizes of $S_{\ell(i)}$ and $S_{r(i)}$. The reason for this is that what matters is the concentration of $|\overline{L}(S_{\ell(i)}) - \overline{L}(S_{r(i)})|$ which is also normalized by the sizes of the two sets.

\begin{figure}[ht!]
    \centering
    \includegraphics[scale=0.38]{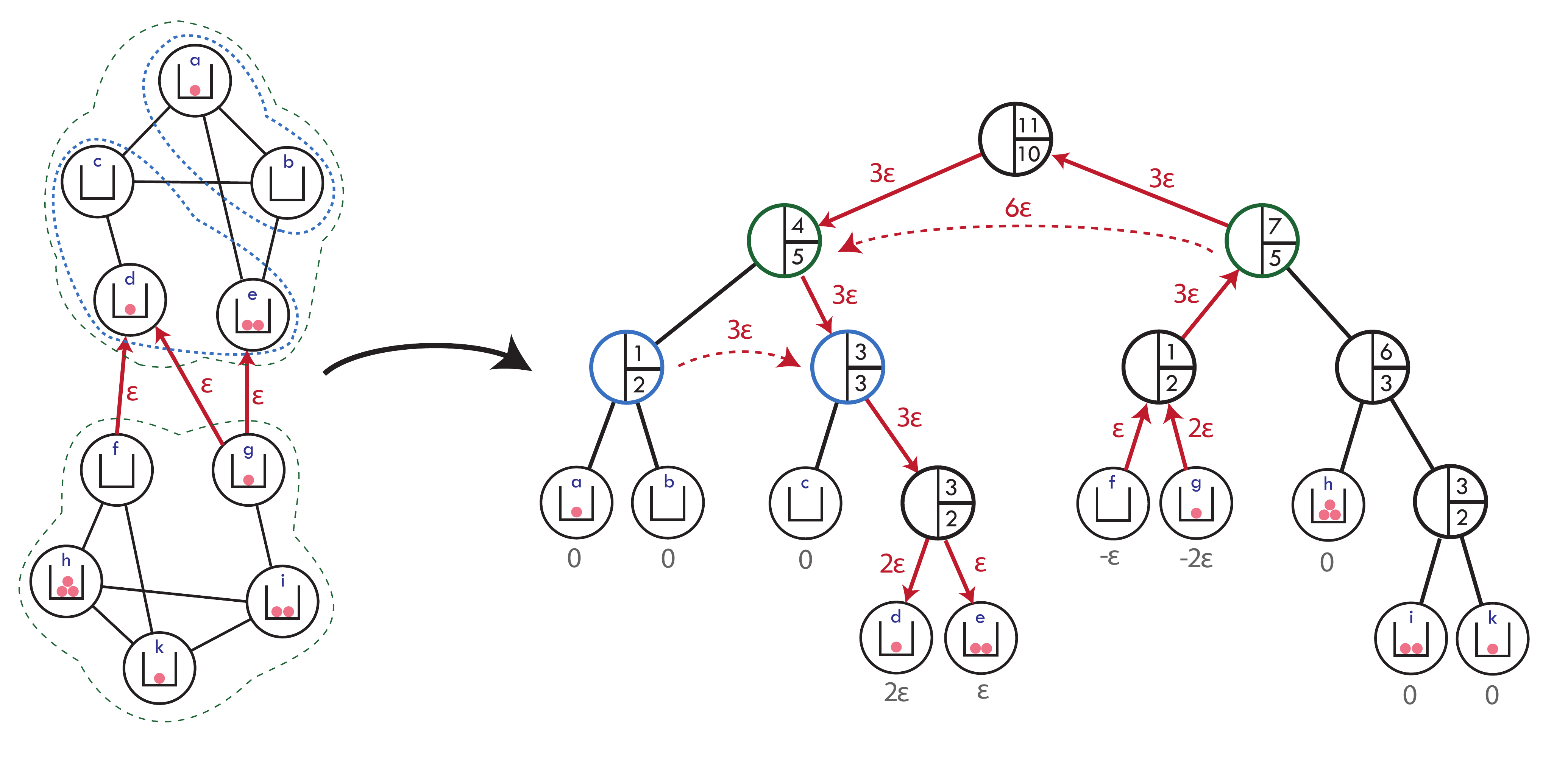}\\
    \putat{-140}{29}{\textcolor{darkgreen}{$S_2$}}
    \putat{-135}{208}{\textcolor{darkgreen}{$S_1$}}
    \putat{-3}{161}{\textcolor{darkgreen}{\scalebox{0.7}{$S_1$}}}
    \putat{131}{161}{\textcolor{darkgreen}{\scalebox{0.7}{$S_2$}}}
    \putat{-167}{188}{\textcolor{orange}{\scalebox{0.7}{$S_{11}$}}}
    \putat{-178}{159}{\textcolor{orange}{\scalebox{0.7}{$S_{12}$}}}
    \putat{-63}{124}{\textcolor{orange}{\scalebox{0.5}{$S_{11}$}}}
    \putat{-2}{124}{\textcolor{orange}{\scalebox{0.5}{$S_{12}$}}}\\[-17pt]
    \includegraphics[scale=0.38]{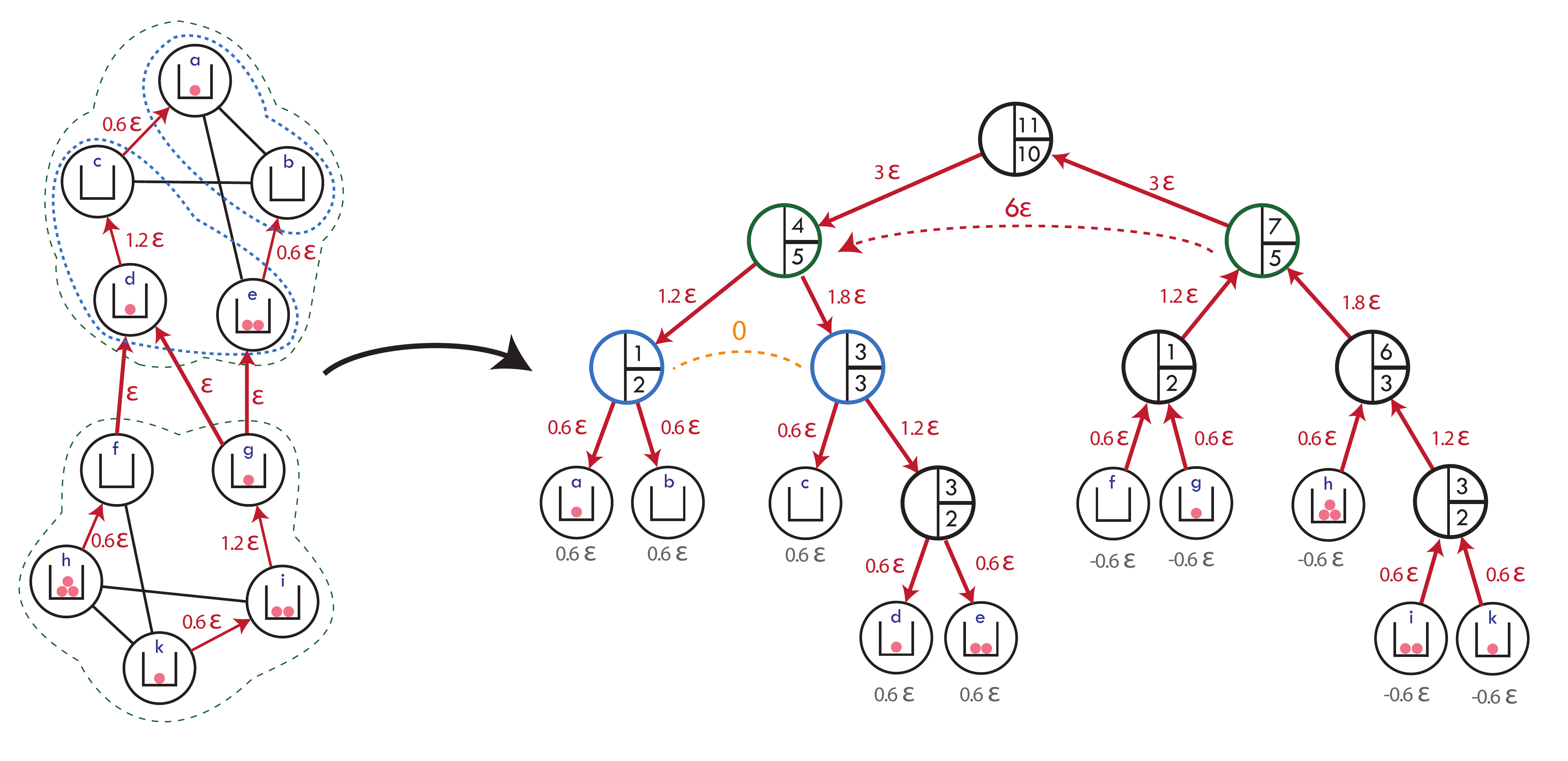}\\
    \putat{-140}{29}{\textcolor{darkgreen}{$S_2$}}
    \putat{-135}{208}{\textcolor{darkgreen}{$S_1$}}
    \putat{-3}{161}{\textcolor{darkgreen}{\scalebox{0.7}{$S_1$}}}
    \putat{131}{161}{\textcolor{darkgreen}{\scalebox{0.7}{$S_2$}}}
    \putat{-167}{188}{\textcolor{orange}{\scalebox{0.7}{$S_{11}$}}}
    \putat{-178}{159}{\textcolor{orange}{\scalebox{0.7}{$S_{12}$}}}
    \putat{-63}{124}{\textcolor{orange}{\scalebox{0.5}{$S_{11}$}}}
    \putat{-2}{124}{\textcolor{orange}{\scalebox{0.5}{$S_{12}$}}}\\[-10pt]  
    
    \caption{The figure shows a $10$ vertex graph $G$ with vertices depicted by bins and the dots in a bin indicating the current load. The weights on the edges specify the edge biases (illustrated by red). The graphs and the vertex loads on the top and the bottom {\em left} are identical, but the edge biases (illustrated by red arrows) are different.
    As $S_2$ has higher average load than $S_1$ ($\overline{L}_{S_1}=7/5$ and $\overline{L}_{S_2}=4/5$), on both {\em top} and {\em bottom}, we create $3 \epsilon$ units of balancing drift from $S_2 $ into $S_1$ (to balance the gap between $S_2$ and $S_1$).
    The drift on {\em bottom} is orthogonal, while the one on {\em top} is not.\\
    On the {\em right} side, we see the resulting drifts on the sets in $G$'s decomposition tree (both trees are identical). 
    Inside a node we depict the average load (as balls divided by bins) for the corresponding set. The drifts on {\em top} create an undesired relative drift from the underloaded set $S_{11}$ to the overloaded set $S_{12}$ (the horizontal dashed line).
    On the other hand, on bottom right, 
    there is only a drift from $S_2$ to $S_1$ and the relative drift between any other pair of siblings is $0$. E.g., the drift entering $S_{11}$ and $S_{12} $ is exactly in proportion to their sizes $2$ and $3$. So the drift does not interfere with balancing $|\overline{L}_{S_{11}} - \overline{L}_{S_{12}}|$.  The same holds for all other pairs of siblings except $S_1$ and $S_2$. 
    \label{fig:edge flow}}
\end{figure}

\paragraph{Using R\"acke Trees to realize the drift.} 
 
To realize the above system of drifts by feasible flows, the final piece is R\"acke trees and the multicommodity flow templates associated with them.

A  R\"acke tree $R$ of a graph $G=(V,E)$ corresponds to a hierarchical decomposition $R$ together with a collection of flow templates $\{f_{uv}\}_{u,v\in V}$ for each pair of vertices $u,v$. Each internal vertex $i\in R$ corresponds to a set $S_i\subset V$, where the leaves of correspond to the singleton vertices of $G$. The capacity of a parent-child edge $(i,j)\in R$ is the capacity of the corresponding cut $(S_j,S_j^c)$ in $G$. The flow templates $\{f_{uv}\}_{u,v\in V}$ specify how to route one unit of flow from $u$ to $v$ in $G$ and have the following remarkable property: let $D=\{\kappa(u,v)\}_{u,v \in V}$ be {\em any} multicommodity demand vector where $\kappa(u,v)$ units of commodity $k_{u,v}$ to be routed from $u$ to $v$. If $D$ can be routed feasibly on $R$ (along the unique $u$ to $v$ path on the tree) subject to the edge capacities in $R$, then the corresponding flow in $G$ produced using the flow templates has load at most $\alpha_G c_e$ on any edge $e \in E$, where $c_e$ is the capacity of $e$. We call $\alpha_G$ the \emph{congestion ratio} of $G$ with respect to $R$.


The best bound on $\alpha_G$ is $O(\log^2 n \log \log n)$ \cite{HHR03} for general graphs, and  the corresponding $R$ and flow templates $f_{uv}$ can be constructed in polynomial time. Almost linear time constructions with slightly worse $\alpha_G$ are also known \cite{RST14}.
The tree $R$ is also  \emph{well-balanced} so that $|R_j|/|R_i|\le \frac34$ for each child $R_j$ of $R_i$, and thus it has depth $O(\log n)$. 



In the preprocessing stage, given $G$ we first construct the R\"acke tree $R$ and modify it into a binary tree $T$.
For each internal node $i$ of $T$, we define the orthogonal drift between the corresponding sets $S_{\ell(i)}$ and $S_{r(i)}$ as described above. Together with the flow templates, which are fixed, this defines {\em fixed} probabilities $p_e(i)$ that each edge $e$ will use to bias 
according to the sign of $\overline{L}(S_{\ell(i)}) - \overline{L}(S_{r(i)})$.
As the drifts are orthogonal and the sets $S_i$ and the flow templates are fixed, this allows us to decouple the drifts for different $i$.

The resulting allocation strategy after preprocessing is extremely simple to describe and efficient to implement and requires only $O(\log n)$ running time per allocation request.

\subsection{Related work and models}
\label{sec:related}
The literature on ball-into-bins processes is extensive and 
it is impossible to cover even a small portion of the developments and applications.
The first appearance of a 2-choice type result was in \cite{KLM96} in the context of online hashing. Following the result of Azar et al.,~\cite{ABKU94} several variations of the model have been studied.
 Variations include models in which balls are eliminated over time \cite{Mthesis91,CFMMSU98} -- either by age or at random and parallel allocation of the balls with limited communication \cite{Stemann96,ACM98}. Many of the earlier results are surveyed in Mitzenmacher's Thesis \cite{Mthesis91} 
and in 
his survey with Richa and Sitaraman \cite{MRS01}. A more recent survey is due to Wieder \cite{Wieder}.

In his thesis  Mitzenmacher suggested the model of $1+\beta$ choice for $\beta<1$, where the algorithm is given two choices with probability $\beta$ and only one choice with probability $1-\beta$. His motivation for introducing this 
model stems from a problem in queuing theory. V\"{o}cking \cite{Vocking03} show that for $d$-choice, non-uniform choices can improve over the greedy algorithm of \cite{ABKU94}, resulting in maximum load of $\Theta(\log\log (n)/d)$ (cf.,~$\Theta(\log\log (n)/\log d)$). 

The heavily loaded case of the two-choice problem was first analyzed in  \cite{BCSV06} (see a neat and short proof by Talwar and Wieder \cite{TW14}). In \cite{PTW15}, Peres, Talwar and Wieder considered the $1+\beta$ choice model for complete graphs and showed that there both gap and upper gap are $\Theta((\log n)/\beta)$. The drift and potential methods introduced in this work 
allowed the authors to relate this result to the graphical case for expanders in \cite{PTW15} and inspired methods used here as well.

\section{Algorithm}
Let $G=(V,E)$ be the underlying graph for the graphical process and let $k = k(G)$ denote the edge-connectivity (the size of the minimum cut) of  $G$.
The algorithm consists of two parts: a preprocessing stage and an online allocation stage. The preprocessing stage takes $G$ as input and computes a binary decomposition tree $T$ of $G$, together with some information on how an edge should bias its allocation based on the imbalance in average load on certain siblings in $T$. 
This information is used to realize a simple online allocation strategy.

Let $I_T$ denote the set of internal nodes of $T$. 
For an internal node $i \in i_T$, let $S_{\ell(i)}$ and $S_{r(i)}$ denote the subsets of $V$ corresponding to the left (resp.~right) child $\ell(i)$ (resp.~$r(i)$) of $i$.
The goal of each node $i\in I_T$ is to control $|\overline{L}^t(S_{\ell(i)})-\overline{L}^t(S_{r(i)})|$, the difference in average load on $S_{\ell(i)}$ and $S_{r(i)}$ at the end of time $t$.

The preprocessing stage also outputs the following (fixed) vectors for each edge $e=(x,y) \in E$. \begin{enumerate}
    \item A probability vector $p_e:I_T\cup \{\emptyset\}\to [0,1]$ specifying the probability that $e$ will determine the allocation based on comparing $\overline L(S_{\ell(i)})$ and $\overline L(S_{r(i)})$.
    \item  A direction vector $\sigma_e: I_T \rightarrow  \pm1$ indicating to which endpoint of $e$ the ball should be assigned if $\overline L(S_{\ell(i)})-\overline L(S_{r(i)})$ is positive.
\end{enumerate}

Before we discuss the preprocessing steps, we first describe the allocation strategy.

\subsection{An online allocation strategy}\label{sec:allocstat}


At each time $t+1$ for $t=0,1,2,\dots,$ the allocation strategy does the following. 

\paragraph{Allocation.} Given a request for an edge $e=(x,y)$: 
\begin{enumerate}
    \item Select an $i \in I_{T}\cup \{\emptyset\}$, according to the precomputed distribution $p_e$.
    
    \item If $i=\emptyset$ allocate the ball randomly to either endpoint of $e$ with probability $1/2$.
    \item Otherwise, consider the quantity $\sigma_e(i)\, (\overline{L}^t(S_{\ell(i)})- \overline{L}^t(S_{r(i)})$.
    \begin{enumerate}
    \item If this is positive, allocate the ball to $y$.\footnote{The vectors $\sigma$ will satisfy $\sigma_{(x,y)}(i) = -\sigma_{(y,x)}(i)$, so the allocation in independent of whether we write $e=(x,y)$ or $e=(y,x)$.}
    \item If this is negative, allocate the ball to $x$.
    \item Otherwise if it is exactly zero,  allocate the ball to an endpoint of $e$  uniformly at random.
    \end{enumerate}
\end{enumerate}

\paragraph{Update.} Update $\overline{L}^t(S)$ for the sets $S$ containing the vertex to which the ball was assigned.

\subsection{Preprocessing}
The preprocessing consists of four stages. In the first stage we compute a R\"acke cut-based decomposition $R=(V_R,E_R,c_R)$ tree for $G$, together with the flow templates $f_{uv}$ for each  ordered pair $u,v \in V$. In the second stage we use $R$ to construct a  binary decomposition tree $T$. In the third stage we compute a balancing drift and associated flows corresponding to each node $ i \in I_T$.
Finally we use these flows to construct the vectors $p_e$ and $\sigma_e$ required to run the allocation strategy.

\subsubsection{Step 1: Construct a R\"acke decomposition.}
Given $G$, compute a R\"acke cut-based decomposition $R=(V_R,E_R,c_R)$ tree for $G$, together with the flow templates $f_{uv}$ for each ordered pair $u,v \in V$, as defined in Section \ref{sec:overview}. This is done using the algorithm of \cite{HHR03} so that congestion ratio of $G$ with respect to $R$ is $\alpha_G =O(\log^2 n \log \log n)$.

\subsubsection{Step 2: From general tree to a binary one.}\label{subs: 3.2.2}

In this step, we convert $R$ into a {\em binary} hierarchical decomposition tree $T = (V_T,E_T)$ of $G$. The nodes of $T$ will define the sets
whose average load we will balance in the allocation strategy. 

\begin{enumerate}
\item Apply the following  step repeatedly to $R$ until all non-leaf nodes have degree $2$.

\item Choose a node $i$ with $p>2$ children $j_1,\ldots,j_p$.
Let $w(h) = |S_{j_h}|/|S_i|$ for $h\in [p]$, so that $\sum_{h=1}^p w(h)=1$.

\begin{enumerate}
\item\label{item: a} If $w(h) > 1/4$ for some $h$, make $j_h$ the left child of $i$. Create a new right child $b$ with the remaining $p-1$ children other than $j_h$. The node $b$ corresponds to the set $S_i \setminus S_{j_h}$.

\item\label{item: b} Otherwise, arbitrarily partition $\{j_1,\ldots,j_h\}$ into two sets $A$ and $B$ so that $w(A),w(B) \in [1/4,3/4]$. Create two children $a$ and $b$ of $i$,
where the children of $a$ (resp.~$b$) are the nodes in $A$ (resp.~$B$). The nodes $a$ and $b$ corresponds to sets $\cup_{j \in A} S_j$ and $\cup_{j \in B} S_j$.
\end{enumerate}
\end{enumerate}
Clearly, $T$ is binary. Also $T$ and $R$  have the same set of leaves as $R$ is a contraction of $T$.

\subsubsection{Step 3: Defining drifts and computing flows.}
Let $I_T$ denote the set of internal (non-leaf) vertices of $T$.
For $i\in I_T$, let $\ell(i)$ and $r(i)$ denote its left and right children in $T$.
Let $\beta = 1/(8 \alpha_G)$ and recall that $k$ is the edge connectivity of $G$. 
For each $i\in I_T$, we create a commodity $i$, and send 
 \begin{equation}
     \label{eq:kappa}
     \kappa_{i}(u,v) =  \frac{\beta k }{|S_{\ell(i)}||S_{r(i)}|}\ind(u\in S_{\ell(i)}, v\in S_{r(i)})
 \end{equation}
units of flow in $T$ from each leaf $u \in S_{\ell(i)}$ to each leaf $v \in S_{r(i)}$.
Let $d_i(z)$ denote the amount of commodity $i$ {\em entering} vertex $z$, so that for each $u \in S_{\ell(i)}$, $v \in S_{r(i)}$ and $w \in V \setminus S_i$, we have
\begin{equation}
    \label{eq:drifts}
    d_i(u)=  -\beta k/|S_{\ell(i)}|,  \qquad d_i(v)= \beta k/|S_{r(i)}|\quad \text{ and  }\quad  d_i(w)=0.
\end{equation}
As we shall see, the allocation strategy on $G$ will precisely produce the drift vector $d_i$ or $-d_i$ for balancing $\overline{L}(S_{\ell(i)})$ and  $\overline{L}(S_{r(i)})$, where the sign depends on which one of these loads is higher.

Consider the flow templates $f_{uv}$ for each pair $u,v \in V$  in $G$, given by the oblivious routing associated with the tree $R$. 
Let $g_i$ be the flow of commodity $i$ produced in $G$ by routing the flows $\kappa_i(u,v)$ using the flow templates $f_{uv}$. That is,
for an edge $e = (x,y)$ of $G$, let
\begin{equation}
    \label{eq:gi}
    g_{i}(x,y) =  \sum_{u,v\in V} f_{uv}(x,y) \kappa_i(u,v).
\end{equation}

Note that the flow $g_{i}$ only depends on the graph $G$ (via $R, T, \{f_{uv}\}_{u,v\in V}$). In particular, it does not depend on the current load vector $L^t$ and hence is fixed over time. We use these $g_{i}(x,y)$ to define $p_e$ and $\sigma_e$.

\subsubsection{Step 4: Defining $p_e$ and $\sigma_e$.}
For each edge $e$, and for each $i\in I_T$, we set \begin{equation}
    \label{eq:pe}
    \sigma_e(i)=\sgn{g_{i}(e)}  \quad \text{ and  }  \quad p_e(i)=|g_{i}(e)|.
\end{equation}
Finally, set $p_e(\emptyset)=1-\sum_{i\in I_T}p_e(i)$.
This concludes the preprocessing\footnote{Note that $\sigma(x,y) = -\sigma(y,x)$ as $g_{i}(x,y) = - g_{i}(y,x)$ (being a flow). This does not affect the allocation strategy as writing  $e=(x,y)$ or $e=(y,x)$ flips both $\sigma_e(i)$ and the $x$ and $y$ in steps 3(a)-3(c).}.


\section{Analysis}\label{sec:anal}

In this section we show  that the preprocessing stage produces $p_e, \sigma_e$ and an allocation strategy with the desired properties and use these to establish Theorem \ref{thm:main}. To do so, we first study the drifts realized by the algorithm and then establish congestion and concentration bounds.

\subsection{Verifying the preprocessing}
We begin by showing that $T$ is well-balanced.
\begin{claim}\label{cl:depth}
\label{lem:ratio}
Let $a, b$ be two nodes in $T$ such that $a$ is an ancestor of $b$, and let $q$ be the distance from $a$ to $b$. Then $|S_b| \leq (3/4)^{\lfloor q/2 \rfloor } |S_a|$. 
In particular, $T$ has depth $O(\log n)$. 
\end{claim}
\begin{proof}
Consider some path $f,g,h$ of length $2$ in $T$, where $f$ is the ancestor of $h$. It suffices to show that either $|S_h|\leq (3/4) |S_g|$ or $|S_g|\leq (3/4) |S_f|$.
If $g$ was a node in $R$, this follows by the well-balancedness of a R\"{a}cke tree for the edge $(g,h)$. Otherwise $g\notin R$, and no matter whether 
 $(f,g)$ was produced according to rule \eqref{item: a} or \eqref{item: b} in Section~\ref{subs: 3.2.2}, we have $|S_g|\le (3/4)|S_f|$.
 \end{proof}

The demand vectors $d_i$ satisfy a useful orthogonality property, as illustrated in Figure~\ref{fig:edge flow}. For $W\subset V$, let $d_i(W)=\sum_{w\in W}d_i(w)$ denote the total flow of commodity $i$ entering the set $W$.
\begin{lemma}
\label{lem:orthogonal}
 For any $i,j \in I_T$, with $i\neq j$ we have that 
 \begin{equation}\label{eq: orth}  \frac{d_{i}(S_{\ell(j)})}{|S_{\ell(j)}|} -     \frac{d_{i}(S_{r(j)})}{|S_{r(j)}|}=0.
 \end{equation}
 For $j=i$ we have
  \begin{equation}
      \label{eq:disi}\frac{d_{i}(S_{\ell(i)})}{|S_{\ell(i)}|} -     \frac{d_{i}(S_{r(i)})}{|S_{r(i)}|}=  - \beta k \left(\frac{1}{|S_{\ell(i)}|} + \frac{1}{|S_{r(i)}|}  \right) .
  \end{equation}
\end{lemma}
 \begin{proof}
Fix $i\in I_T$ and let $T_i$ denote the subtree of $I_T$ rooted at $i$. We first consider $j\neq i$. 
Recall that the flow of commodity $i$, satisfying \eqref{eq:kappa} and \eqref{eq:drifts},  is only between nodes of $S_i$ and that $d_i(w)=0$ for $w \notin S_i$.
So if $j \notin T_i$ and $i \notin T_j$, then \eqref{eq: orth} holds trivially as $S_j \cap S_i=\emptyset$  and  $d_{i}(S_{\ell(j)}) =d_{i}(S_{r(j)}) =0$.  

Next, if $i \in T_j$ (and $i\neq j)$ then $S_i$ is
contained entirely in either $S_{\ell(j)}$ or $S_{\ell(i)}$. As the total commodity $i$ flow  leaving $S_i$ is $0$, this gives that  $d_{i}(S_{\ell(j)}) =d_{i}(S_{r(j)}) =0$. 

Next, suppose $j \in T_i$ (and $j\neq i$) and assume without loss of generality that $j$ lies in the right subtree of $i$. Then positive flow enters each vertex of $S_j$ and the flow entering $S_{\ell(j)}$ is exactly 
 \begin{equation}
\label{eq:demij}
     d_i(S_{\ell(j)}) =  \sum_{u \in S_{\ell(j)}} 
     d_i(u) = 
     \beta k \frac{ |S_{\ell(j)}|} {|S_{r(i)}|}.
 \end{equation}
 Similarly $d_i(S_{r(j)}) =  \beta k|S_{r(j)}| /|S_{r(i)}|$, 
 and the lemma follows. 
 Likewise, if $j$ lies in the left subtree of $i$ (so that the flow leaves $S_{\ell(j)}$ and $S_{r(j)}$), the same argument holds, up to a change of sign.

 Finally for $j=i$, we have $d_i(S_{\ell(j)}) =  -\beta k$ and $d_i(S_{r(j)})=\beta k$, and the lemma follows.
 \end{proof}

\begin{lemma}
\label{lem:tot-demand}
For any node $j\in I_T$ we have 
$\sum_{i\in I_T}|d_i(S_j)|\le 8 \beta k=  k/\alpha_G$.
\end{lemma}
\begin{proof}
Fix a node $j \in I_T$. 
For any $i$ that lies in the subtree $T_j$ rooted at $j$,
we have that $d_i(S_j)=0$ as $d_i(u)=0$ for all $u\notin S_i$, $d_i(S_i)=0$ and  $S_i \subseteq  S_j$.


Next, consider any $i$ such that $j \notin T_i$ (and also $i \notin T_j$ by the case above), then $S_j \cap S_i = \emptyset$ and we have that $d_i(S_j)=0$ as $d_i(w)=0$ for $w \notin S_i$.
 
Thus it suffices to consider nodes $i$ which are ancestors of $j$.
Suppose, without loss of generality, that $j$ is in the left subtree of $i$. 
Then, by Claim~\ref{lem:ratio} and \eqref{eq:demij}, 
\[
|d_i(S_j)| = \Big|\sum_{u \in S_j} d_i(u)\Big| \le \beta k |S_j|/|S_{\ell(i)}| \leq \beta k (3/4)^{\lfloor (\dist(i,j)-1)/2 \rfloor},
\]
where $\dist(i,j)$ is the distance between $i$ and $j$.
The result follows by summing up over all the ancestors $i$ of $j$ and using that $(3/4)^{1/2} \leq 7/8$.
\end{proof}

We  now show that $p_e$ is a valid probability vector.
\begin{lemma}
\label{lem:R-load-n}
For every $e\in E$ of the graph $G$, we have that $\sum_{i\in I_T} p_e(i)\le 1$.
\end{lemma}
\begin{proof}
\allowdisplaybreaks
Fix an edge $e=(x,y)$ of $G$.
Recall that for each node $i\in I_T$, we defined $p_i(e)=|g_i(x,y)|$ in \eqref{eq:pe} where the flow $g_i$ is obtained by mapping the commodity $i$ flow in $R$ to $G$ using the flow templates as defined in \eqref{eq:gi}.
So $\sum_{i\in I_T} p_i(e)$ is exactly the total multi-commodity flow through the edge $e$.

By the property of R\"acke trees, it suffices to show that the total multi-commodity flow through any edge $(j,j')$ of $R$ is at most $1/\alpha_G$ times its capacity.
Moreover as $G$ is $k$ edge-connected, each edge in $R$ has capacity at least $k$, and it suffices to show that this is at most $k/\alpha_G$.

Let $j'$ be the child of $j$ in $R$. Then the edge $(j,j')$ corresponds to path in $T$ with $j$ an ancestor of $j'$. 
The flow across $(j,j')$ in $R$ of commodity $i$ is due to demands $\kappa_i(u,v)$ for all pairs $(u,v)$ with exactly one of $u$ or $v$ is contained in $S_j$, i.e., $|\{u,v\} \cap S_j|=1$.
For each $i \in I_T$, the total commodity $i$ flow that leaves or enters the set  $S_j$ is exactly $|d_i(S_j)|$ and hence by Lemma \ref{lem:tot-demand}, the total flow across $(j,j')$ is at most $k/\alpha_G$.
\end{proof}

\subsection{Realized flows and drifts}
The vectors $p_e$, $\sigma_e(i)$ are constructed so that the allocation strategy at each time achieves a certain drift vector.
Fix a time $t$. Given a load vector $L^t$, for each $i\in i_T$ define
\[Q_{L^{t}}(i) =\begin{cases}
 +1    &  \text{ if } \overline{L}^t(S_{\ell(i)}) >  \overline{L}^t(S_{r(i)}), \\
 -1    &  \text{ if } \overline{L}^t(S_{\ell(i)}) <  \overline{L}^t(S_{r(i)}), \\
 0  & \text{ otherwise.} 
\end{cases}
\]

Notice that using these $Q_{L^t}(i)$ our allocation strategy at time $t+1$ as described in Section \ref{sec:allocstat}
 can be expressed more compactly as follows. When edge $e=(x,y)$ is requested, with probability $p_e(i)$, assign the ball to $y$  with probability $(1+\sigma_e(i)Q_{L^{t}}(i))/2$ and to $x$ otherwise. 
In other words, the strategy creates a bias of $\sigma_e(i) Q_{L^t(i)}$ with probability $p_e(i)$, towards $y$ upon request to $(x,y)$ at time $t+1$.

Summing up over $i$, the overall bias towards $y$  upon the request to edge $e=(x,y)$ is 
\begin{equation}
\label{eq:bias}
b_t(x,y) = \sum_{i\in I_T}Q_{L^t}(i)\sigma_e(i)p_e(i) = 
\sum_{i\in I_T}\sgn \big(g_i(x,y)\big)\left|g_i(x,y)\right| Q_{L^t}(i)=\sum_{i\in I_T}g_i(x,y) Q_{L^t}(i),
\end{equation}
where the second equality uses the definition of $\sigma_e(i)$ and $p_e(i)$ in \eqref{eq:pe}.

In other words, the edge biases $b_t(x,y)$ are exactly realized by the sum of flows $g_i(x,y)$ weighted by the
signs $Q_{L^t(i)}$ which are based on the current loads.

Using this we can now describe the probability $q_t(y)$ that a vertex $y \in V$ is assigned a ball at step $t+1$.
Henceforth, we will focus on a fixed time and drop the dependence on $t$ in $L^t, \overline{L}^t, Q_{L^t}, b_t(x,y)$ and $q_t(y)$ for ease of notation.  These quantities should  always be viewed as functions of $t$.

\begin{lemma}\label{lem:comp dsec}
At every time $t$, writing $L=L^t$, for each vertex $y$ of $G$ the allocation strategy assigns a ball to $y$ at time $t+1$ with probability 
\[ q(y) = \frac{1}{n} + \frac{1}{2m} \sum_{i\in I_T}d_i(y) Q_L(i)
.\]
Equivalently, the drift $d^t(y) = 2m (q(y)-1/n)$ at time $t$ at each node $y$ is exactly $\sum_{i\in I_T}d_i(y) Q_L(i)$.
\end{lemma}
\begin{proof}
As each edge $e$ is requested uniformly at random, by \eqref{eq:bias} the probability of allocating the ball to $y$ is exactly 
\[  q(y) = \frac{1}{m}  \sum_{x \in N(y)} \frac{1+b(x,y)}{2} =  \frac{d}{2m} + 
\frac{1}{2m} \sum_{x \in N(y)}\sum_{i\in I_T}g_i(x,y) Q_{L}(i).\]
As $2m=nd$ and $\sum_{x \in N(y)} g_i(x,y)$ is the commodity $i$ flow entering $y$,  which by definition  is exactly $d_i(y)$, we obtain that 
\[ q(y) = \frac{1}{n} +  \frac{1}{2m} \sum_{i\in I_T}d_i(y) Q_L(i). \qedhere \]
\end{proof}

 For a subset of vertices $W\subset V$, let $q_t(W)=\sum_{w\in W}q_t(w)$ (denoted $q(W)$ henceforth)  denote the probability of allocating a ball at time $t+1$ to some vertex in $W$. 
 
 We can now compute the relative drift between any two sibling sets $S_{\ell(i)}$ and $S_{r(i)}$. 
\begin{lemma}
\label{cor:qsi-normalized-diff}
For a node $i\in I_T$ and its children $\ell(i)$ and $r(i)$,
\[ \frac{q(S_{\ell(i)})}{|S_{\ell(i)}|} - \frac{q(S_{r(i)})}{|S_{r(i)}|} = \frac {-\beta k Q_L(i)}{2m} \left(\frac{1}{|S_{\ell(i)}|}+   \frac{1}{|S_{r(i)}|}\right).
\]
\end{lemma}
Observe that this only depends on $Q_L{(i)}$, i.e.~on how they average loads on $S_{\ell(i)}$ and $S_{r(i)}$ compare, and not on anything else (this is exactly where the orthogonality property is useful). This will allow us to establish concentration bounds on the difference $|\overline{L}(S_{\ell(i)}) - \overline{L}(S_{r(i)})|$.

\begin{proof}
By Lemma~\ref{lem:comp dsec}, we have  that 
\[q(S_{\ell(i)}) = \sum_{y \in S_{\ell(i)}} q(y) = |S_{\ell(i)}|/n + \sum_{j\in i_T} d_j(S_{\ell(i)}) Q_L(j)/2m,\] and  an analogous expression for $q(S_{r(i)})$ with $S_{\ell(i)}$ replaced by $S_{r(i)}$. Together this gives,
\[   \frac{q(S_{\ell(i)})}{|S_{\ell(i)}|} - \frac{q(S_{r(i)})}{|S_{r(i)}|} = \frac{1}{2m}  \sum_{j \in i_T }  Q_L(j) \left(\frac{d_j(S_{\ell(i)})}{|S_{\ell(i)}|}  -    \frac{d_j(S_{r(i)})}{|S_{r(i)}|}\right). \]
By the orthogonality property in Lemma~\ref{lem:orthogonal}, the only non-zero contribution in the sum above is due to the term $j=i$. Using \eqref{eq:disi} the result follows.
\end{proof}

\subsection{Controlling the gap through concentration}

With the results above, we are now ready to prove Theorem~\ref{thm:main}. To this end, we first state a bound on the gap between the average load on sets corresponding to any two siblings in $T$. 
Recall that $d$ and $k$ are the degree and edge-connectivity of $G$.
\begin{lemma}
\label{lem:main}
Writing $\overline{L}=\overline{L}^t$, for any time $t$ and any non-leaf node $i \in I_T$, for every constant $c>0$,
\[ \Pr \left[\,|\overline{L}(S_{\ell(i)}) -  \overline{L}(S_{r(i)})|  > 8( d/k) \alpha_G \, c \log n \right] = O(n^{-c}).\]
\end{lemma}
Let us first see how this directly implies Theorem \ref{thm:main}.
\begin{proof}[Proof of Theorem \ref{thm:main}]
Fix any non-leaf node $i \in I_T$ and consider some child $j \in  \{\ell(i), r(i)\}$ of $i$. Let $j'$ be the sibling of $j$.
We first observe that the difference between average load of siblings is no less than the difference for a parent-child pair. This follows as
\begin{equation}
    \label{eq:rela}
    \left|\overline{L}(S_j) - \overline{L}(S_{j'})\right|
= \left|\frac{L(S_j)}{|S_j|} - \frac{L(S_i)-L(S_j)}{|S_{j'}|}\right|
= \frac{|S_i|}{|S_{j'}|} \left| \overline{L}(S_j) -  \overline{L}(S_i)\right| \geq  \left| \overline{L}(S_j) -  \overline{L}(S_i)\right|,
\end{equation} 
where we use the fact that $L(S_i) = L(S_j)+L(S_{j'})$, $|S_i|=|S_j| +|S_{j'}|$ and $|S_i| \geq |S_{j'}|$.

 Fix some vertex $u$ of $G$, and consider the path $u=i_0,i_1,\ldots,i_h=r$  in $T$ from  the leaf $u$ to the root $r$. 
 As $L(u) = \overline{L}(u)$ for a leaf $u$ and $\overline{L}(r)$ is the average load over $V$,
 the deviation of the load $u$ from the global average is  
 \[ \left|\overline{L}(u) - \overline{L}(r)\right| 
 = \left|\sum_{g=1}^h \big(\overline{L}(S_{i_{g-1}}) - \overline{L}(S_{i_{g}})\big)\right| \leq\sum_{g=1}^h \left|\overline{L}(S_{i_{g-1}}) - \overline{L}(S_{i_{g}})\right| .\] 
Applying Lemma \ref{lem:main} with $c>1$, taking  
a union bound over the $O(n)$ edges $(i,j) \in E_T$ and using \eqref{eq:rela}, we get that w.h.p. $|\overline{L}(S_i) - \overline{L}(S_j)|\le ( d/k) \alpha_G \, c' \log n$ for all edges $(i,j) \in E_T$.  As the height of $T$ is $h_T = O(\log n)$ by Claim~\ref{cl:depth}, 
this gives that,  w.h.p., for every vertex $u$,
\[\left|\overline{L}(u) -  \overline{L}(r)\right|  = O(( d/k) \alpha_G \, \log^2 n).\]
Recalling that $\alpha_G = O(\log^2 n \log \log n)$ gives the claimed result.
\end{proof}

\subsection{Proof of the concentration estimate}
We  now prove Lemma \ref{lem:main}.
We first develop a concentration lemma to be used in our analysis. 
\paragraph{2-point concentration.} Consider the following set up. There are two bins, $1$ and $2$, associated with
\emph{steady state probabilities} $\pi_1, \pi_2$ which satisfy $\pi_1+\pi_2=1$. 
At each time step a ball arrives and
let $\ell^t_1, \ell^t_2$ denote the loads of the bins at the end of time $t$. In addition, there is  fixed parameter $\varepsilon\leq \min(\pi_1,\pi_2)/2$.

Suppose the allocation process at $t+1$ is as follows:
choose a parameter $\varepsilon_t>\varepsilon$  (independent of future allocations, but possibly depending on the history of the process and external randomness) which almost surely satisfies $\varepsilon_t\le \max(\pi_1,\pi_2)/2$.
If $\ell^t_1/\pi_1 > \ell^t_2/\pi_2$, then allocate the ball to bin 1 with probability $\pi_1-\varepsilon_t$ or to bin 2 otherwise. If $\ell^t_1/\pi_1 > \ell^t_2/\pi_2$,  allocate the ball to bin 1 with probability $\pi_1 + \varepsilon_t$ or to bin 2 otherwise. Finally, if $\ell^t_1/\pi_1 = \ell^t_2/\pi_2$ allocate the ball to bin 1 with probability $\pi_1$ or to bin 2 otherwise.

For any time $t$, let
$\delta(t) = \ell^t_1/\pi_1 - \ell^t_2/\pi_2$ denote the normalized difference in the loads of the bin. Then we have the following concentration bound for this 2-point process.
\begin{lemma}
\label{lem:2pt}
For any time $t$ and any $x\geq 0$ it holds that 
\[\Pr[|\delta(t)| \geq  x/\varepsilon] = O\big(\exp(-x/8)\big).\]
\end{lemma}
\begin{proof}
We define the potential function  
$\Phi(t) =  \cosh (\alpha \delta(t))$, where  $\alpha = \varepsilon/8$.
So $\Phi(0) =1$ initially at $t=0$.  

Fix a time $t$ and let $\Delta \Phi = \Phi(t+1) - \Phi(t)$ and denote $\delta=\delta(t)$ and  $\Delta \delta = \delta(t+1)  - \delta(t)$.
Then, by Taylor expansion, and using $(d/dx) \cosh(x) = \sinh(x)$ and $(d/dx) \sinh x = \cosh x$,  
\[\Delta \Phi = \sinh (\alpha \delta) \Big(\alpha \Delta \delta + (\alpha \Delta \delta)^3/3! + \ldots\Big)  + \cosh (\alpha \delta)  \Big(\alpha (\Delta \delta)^2/2! + (\alpha \Delta \delta)^4/4! + \ldots\Big). \]
Let $\gamma =1/\pi_1  + 1/\pi_2$. As $\varepsilon \leq \max(\pi_1,\pi_2)/2$ we have that $\gamma \varepsilon \leq 1$. As either $\ell_1^t$ or $\ell_2^t$ rises by $1$, we have $|\Delta \delta| \leq \gamma$, and hence $|\alpha \Delta \delta| \leq 1/2$.
Bounding $|(\alpha \Delta \delta)^i | \leq (\alpha \Delta \delta)^2 2^{-i+2}$ for $i\geq 2$, and using that $ \cosh x -1 \leq |\sinh(x)| \leq \cosh x$ for all $x$,
\[  \Delta \Phi \leq \sinh (\alpha \delta) (\alpha \Delta \delta) + \cosh(\alpha \delta) (\alpha \Delta \delta)^2 .\] 
We now bound $\E[\Delta \delta]$ and $\E[(\Delta \delta)^2]$. For $\delta>0$,
\[\E (\Delta \delta)  =  (\pi_1 - \varepsilon_t)\frac{1}{\pi_1} - (\pi_2 + \varepsilon_t)\frac{1}{\pi_2}   =  -\varepsilon_t \left(\frac{1}{\pi_1} + \frac{1}{\pi_2}\right) = -\varepsilon_t\gamma \le -\varepsilon \gamma,\]
while for $\delta< 0$, we have $\E[\Delta \delta] = \varepsilon_t \gamma\ge \varepsilon \gamma$. For $\delta=0$ we have $\E[\Delta \delta]=0$.

For $\delta>0$, we have
\[\E[(\Delta \delta)^2] =  (\pi_1 - \varepsilon_t) \left(\frac{1}{\pi_1}\right)^2 + (\pi_2+ \varepsilon_t)\left(\frac{1}{\pi_2}\right)^2 = \gamma  + \varepsilon_t(-1/\pi_1^2 + 1/\pi_2^2) \leq \gamma + \varepsilon_t \gamma^2 \leq 2 \gamma,\]
where the last step uses the fact that $\varepsilon_t \gamma \leq 1$. The same bound also holds for $\delta< 0$ and for $\delta=0$.

This gives
\begin{eqnarray*} \E[\Delta \Phi(t)\ |\ \Phi(t)] & \leq &  \alpha \sinh (\alpha \delta) \E[\Delta \delta] + 2 \alpha^2 \cosh(\alpha \delta) \E[(\Delta \delta)^2] \\
& \leq  &  - \alpha |\sinh (\alpha \delta) | \varepsilon \gamma + 4 \alpha^2 \gamma \cosh(\alpha \delta) \\
& \leq & -\alpha \varepsilon \gamma (\Phi(t)-1) + \alpha \varepsilon \gamma \Phi(t)/2 = -\alpha \varepsilon \gamma (\Phi(t) -2) /2.
\end{eqnarray*}
Thus, taking expectation and using $\Delta(\Phi(t))=\Phi(t+1)-\Phi(t)$ we obtain
\[\E[\Phi(t+1)]\le \alpha \varepsilon \gamma +(1-\alpha \varepsilon \gamma/2)\E[\Phi(t)].\]
Taking induction over $t$ and recalling that $\Phi(0)=1$ we obtain 
\[\E[\Phi(t)]\le 2\big(1-(1-\tfrac{\alpha \varepsilon\gamma}{2})^t\big)+(1-\tfrac{\alpha \varepsilon\gamma}{2})^t\le 2.\]
The result now follows by applying Markov's inequality to $\Phi(t)$.
\end{proof}

\paragraph{Proof of Lemma \ref{lem:main}.}
We apply the set up above to bound $\left|\ovl(S_{\ell(i)})-\ovl(S_{r(i)})\right|$. Letting bin 1 correspond to $S_{\ell(i)}$ and bin 2 to $S_{r(i)}$, we restrict ourselves to the time steps when the ball is allocated to some vertex in $S_i$. Let $\pi_1 = |S_{\ell(i)}|/|S_i|$ and $\pi_2 = |S_{r(i)}|/|S_i|$ so that $\pi_1+\pi_2=1$. 

By definition of $q=q_t$, the probability of allocating the ball to $S_{\ell(i)}$ conditioned on it being allocated to $S_i$ is $q(S_{\ell(i)})/q(S_i)$, and similarly $q(S_{r(i)})/q(S_i)$ for $S_{r(i)}$. 

Let us write $\pi_1 -\varepsilon_t = q(S_{\ell(i)})/q(S_i)$ and $\pi_2 + \varepsilon_t = q(S_{r(i)})/q(S_i)$.  Dividing the first and second equations by $\pi_1$ and $\pi_2$ and subtracting the first from the second give that 
\[ \varepsilon_t \left(\frac1{\pi_1}+\frac1{\pi_2}\right)  = \frac{ q(S_{r(i)})}{ \pi_2 q(S_i)}  - \frac{ q(S_{\ell(i)})}{\pi_1 q(S_i)} = \frac{|S_i|}{q(S_i)} \left( \frac{ q(S_{r(i)})}{ |S_{r(i)}|} -   \frac {q(S_{\ell(i)})}{ |S_{\ell(i)}|}    \right),  \]
which by Lemma~\ref{cor:qsi-normalized-diff}  equals
\[\frac{|S_i|}{q(S_i)} \frac{\beta k Q_L(i)}{2m} \left(\frac{1}{|S_{\ell(i)}|} +  \frac{1}{|S_{r(i)}|} \right) =  \frac{1}{q(S_i)} \frac{\beta k Q_L(i)}{2m} \left(\frac{1}{\pi_1} +  \frac{1}{\pi_2} \right).\]
This yields  $\varepsilon_t =  \beta k Q_L(i)/ q(S_i)$. 

Since $Q_L(i) \in \{-1,1\}$ and $q(S_i) \leq 2|S_i|/n$ (as $q(v)\leq d/m =2/n$ for any vertex $v$), we deduce that
$|\varepsilon_t| \geq \beta k n/(2m|S_i|)$.
We thus define 
\begin{equation}\label{eq:bdef}
\varepsilon := \frac{\beta k n}{2m |S_i|}=\frac{k}{8d \alpha_G |S_i|},\end{equation}
so that $\varepsilon \leq |\varepsilon_t|$.
Applying Lemma \ref{lem:2pt}, we obtain
\[ \Pr\left[ \left| \frac{L(S_{\ell(i)})}{\pi_1} - \frac{L(S_{r(i)})}{\pi_2}  \right| \geq x/\varepsilon \right] = O\big(\exp(-x/8)\big).\]
As $\pi_1 = |S_{\ell(i)}|/|S_i|,  \pi_2 = |S_{r(i)}|/|S_i|$, and using the value of $\varepsilon$ in  \eqref{eq:bdef}, this becomes
\[  \Pr\left[ \left| \ovl(S_{\ell(i)}) - \ovl(S_{r(i)})  \right| \geq  8x(d/k) \alpha_G\right] = O\big(\exp(-x/8)\big),\]
which implies the lemma. \qed

\subsection{Implementation}
\label{s:implement}
To implement the allocation strategy, the algorithm can maintain a dynamic table of size $O(n)$ containing the load on each set $S$  in the decomposition.
At any time, upon an edge request $e=(u,v)$, the algorithm samples $i \in I_T \cup \{\emptyset\}$ with the probability $p_e(i)$ and looks up the two table entries containing loads on $S_{\ell(i)}$ and $S_{r(i)}$. When a ball is assigned to some vertex $v$
only the $O(\log n)$ entries for $S$ with $v \in S$ are updated. 

\begin{remark} In fact, the number of updates to the table per time step can be reduced to $O(1)$. 
In the distributed implementation part below we show that the bound in Theorem \ref{thm:main} holds even if at any time the entry for the total load of any set $S$ is not very accurate and can differ from the actual load on $S$ by $O(|S| \log n)$.
We can use this slack to let each vertex $v$ make a batch update to the entry for each $S$  with $v \in S$ only every $O(n \log n)$ steps.
As each vertex lies in $O(\log n)$ sets, the total number of batch updates is $O(n \log n)$ per $O(n \log n)$ steps, implying $O(1)$ amortized update time per allocation step.
\end{remark}

Finally, storing the vectors $p_e$ and $\sigma_e$ does not require much memory either, as 
for each edge $e$, the vectors $p_e$ and $\sigma(e)$ have support $O(\log n)$. This follows by the property of R\"acke trees and the flow templates that an edge $e$ is used in $f_{uv}$ only if both endpoints of $e$ lie in $S_{a(u,v)}$ and $S_{a(u,v)}$ where $a(u,v)$ is the least common ancestor of $u$ and $v$ in $T$.
As $T$ has depth $O(\log n)$, for each edge $e$, there are only $O(\log n)$ relevant sets $S_i$ with $p_e(i)\neq 0$. So the total memory used to store these (static) vectors is $O(m \log n)$.

\paragraph{Distributed implementation.}
The implementation above assumed  that each vertex had read/write access to a centralized table containing the loads $L(S)$.
We now consider a distributed setting where 
vertex maintains its own local table with possibly {\em stale} entries $L(S)$ and 
give  a distributed implementation that requires only $O(1)$ amortized messages per allocation, with $O(\log n)$ bits per message. The key underlying observation is  that allocation strategy is quite robust and does not require exact knowledge of the load $L(S)$ on the sets $S$ in the decomposition. In particular, the bound in Theorem \ref{thm:main} also holds when $L(S)$ is inaccurate up to $O(|S| \log n)$ for each set $S$.

More precisely, suppose for any set $S_i$, the allocation strategy is guaranteed to produces the balancing drift in the correct direction only when $|\overline{L}(S_{\ell(i)}) - \overline{L}(S_{r(i)}| \geq c \log n$ and otherwise the direction can be arbitrary, i.e.,~if $|\overline{L}(S_{\ell(i)}) - \overline{L}(S_{r(i)}|  \leq \log n$, then the sign $Q_L(i)$ can be arbitrary.
This affects the gap  $|\overline{L}(S_{\ell(i)}) - \overline{L}(S_{r(i)}|$ in the 2-point concentration argument by only an additive $c \log n$ so that we now have 
\[ \Pr\big[|\overline{L}(S_{\ell(i)}) - \overline{L}(S_{r(i)})| \geq  8 x (d/k)\alpha_G + c \log n \big] = O\big(\exp(-x/8)\big).\]
As $(d/k)\alpha_G \geq 1$ this gives a similar tail bound as the original allocation strategy whenever $x  = \Omega(\log n)$.

So
each vertex can maintain its own
stale estimate of $L(S)$ with error up to $O(|S|\log n)$.
Moreover,
as any set $S$ has at most $2|S|/n$ probability of being assigned a ball at ant time, w.h.p., it suffices for  a vertex needs to update its load entry for $S$ in only $O(n \log n)$ steps.
Let us focus on a fixed set $S$ in the decomposition.
The R\"acke decomposition ensures that the induced graph $G[S]$ on each such set $S$ is connected.
So after every interval of $O(n \log n)$ time steps, the vertices in $S$ can aggregate the information on how many balls in total were assigned to $S$ during that interval. As $S$ is connected, this  can 
be implemented directly using $O(|S|)$ messages of length $O(\log n)$ in total.

As the sum of the sizes of all the sets at a given level in the decomposition tree $T$ is at most $n$ and as $T$ has depth $O(\log n)$, the total number of messages sent during an interval of size $O(n \log n)$ is $O(n \log n)$. This implies an amortized communication overhead of $O(1)$ message per time step.

\section{Lower Bounds}
\label{sec:lb}
We describe various simple but instructive lower bounds. First we prove Theorem \ref{thm:lb} in Section \ref{sec:lb1}. Then, in Section \ref{sec:lb2} we show an $\Omega(\log(n)/\log \log n)$ lower bound on the upper gap for bounded degree graphs. In particular this implies that on bounded degree expanders, the upper gap is roughly of the same order as  the gap, in contrast with the complete graph where the upper gap is $O(\log \log n)$ and the gap is $\Omega(\log n)$. 

We also remark that \cite{PTW15} showed that for complete graphs, under the $1+\beta$ choice model for a fixed $\beta<1$, the upper gap is $\Omega(\log n)$ and hence similar to the gap. Again this is in contrast to $\beta=1$ (the 2-choice model) where the upper gap  $O(\log \log n)$ and gap is $\Omega(\log n)$. 

\subsection{Proof of Theorem \ref{thm:lb}}
\label{sec:lb1}
We show for any $d$-regular $k$ edge-connected graph $G$, under any allocation strategy, the gap is $\Omega(d/k +\log n)$  with at least constant probability.

The $\Omega(\log n)$ bound follows from the following folklore argument. Fix any time $t$ and 
consider any interval $I$ of $O(n \log n)$ steps just before $t$. For a fixed vertex, at any time step, the probability that some edge incident to it is chosen is $d/m = 2/n$. Hence, by standard coupon-collector argument, with $\Omega(1)$ probability, there is some vertex $v$ for which no incident edge is chosen during $I$. So during $I$, the load of $v$ cannot change under any strategy, while the average load increases by $\log n$. 

Hence, to prove Theorem \ref{thm:lb}, it suffices to show the following.
\begin{lemma}
Let $G$ be any $d$-regular, $k$ edge-connected graph. Then for any allocation strategy, and at any time $t$ the gap is $\Omega(d/k)$ with constant probability.
\end{lemma} 
\begin{proof}
Let $C=(S,\overline{S})$ be some minimum cut in $G$, and let $S$ be the larger side with $|S|\geq n/2$.
Fix a time $t$ and 
consider some time interval $I$ of length $T$ just before $t$. We specify $T$ later. Let $Y$ be the number of edges requested with both endpoints in $S$.
As an edge is requested uniformly at random at each time, and as there are $(d|S|-k)/2$ such edges, we obtain
\[ \E[Y] = (d|S|-k)/2 \cdot T/m =  (|S|-k/d)T/n.\]
As $k \leq d$ and  $|S|\geq n/2$, this at least $T/3$, provided that $n\geq 6$.

By standard estimates, with at least constant positive probability, we have $Y \geq  \E[Y] + \sqrt{T}$.
Conditioned on this event, the average load on vertices in $S$ during $I$ exceeds the stationary load of $T/n$ by at least
\[  \frac{Y}{|S|} - \frac{T}{n} \geq  \frac{\E[Y] + \sqrt{T}}{|S|} - \frac{T}{n} =   \frac{(|S|d-k)T}{|S|dn} + \frac{\sqrt{T}}{|S|} -  \frac{T}{n} = \frac{\sqrt{T}}{|S|}  - \frac{k T}{nd|S|}.\]
Choosing $T = cn^2d^2/k^2$ for small enough $c$ and as $|S|\geq n/2$, this is
$\Omega( d/k)$.
\end{proof}
\subsection{Upper gap for bounded degree graphs}
\label{sec:lb2}
Next we show that even for bounded degree expanders, under the greedy strategy the maximum load at time $t$ is typically $t/n + \tilde{\Omega}(\log n)$. This is unlike for complete graphs where the upper gap is $O(\log \log n)$.
\begin{lemma}
For any $d$-regular graph $G=(V,E)$ with $d=O(1)$, at any time $t$ and for any allocation strategy, with constant probability the maximum vertex load is $t/n + \Omega(\log n/\log \log n)$. 
\end{lemma}
\begin{proof} Fix a time $t$, and consider the interval $I$ of length $|E|=dn/2$ just before $t$. Let $t'=t-|I|$ be time at the beginning of $I$, and let $a = t'/n$ denote the average load at $t'$. Let $s=\log n/(4\log\log n)$.
We will show that with constant probability, either $t$ or $t'$ has upper gap $\tilde{\Omega}(\log n)$.

If some vertex has load $a+s/4$ at time $t'$, then we are already done. Otherwise, by an averaging argument, there can be at most $n/4$ vertices with load at most $ a-s$ at time $t'$. Call an edge {\em bad} if it is incident to such a vertex, and let $B$ be the set of bad edges.
As $G$ is $d$-regular, $|B| \leq d(n/4) \leq |E|/2$.

As $|I|=|E|$, and each edge is requested uniformly and independently during $I$, with constant probability, some edge $e\in E\setminus B$ will be requested at least $4s$ times. If this event occurs, then for any allocation strategy, the load on some endpoint of $e$ increases by at least $2s$ during $I$, and hence becomes at least $a-s+2s=a+s$ at time $t$. On the other hand, the average load only increases by $|I|/n=d/2 = O(1)$, to become $a+d/2$ at time $t$. This implies the result.
\end{proof}

\subsection*{Concluding Remarks}
Our result gives, in particular, an allocation strategy with asymptotically polylogarithmic gap for any $d$-regular $d$-connected graph. One may wonder whether this strategy gives any improvement in practical settings where the number of bins is not too large. To show that it indeed does, we conclude the paper with with a simulation result of our strategy for cycles, compared against the asymptotic behavior of the greedy strategy (Figure~\ref{fig:Flow Gap}).

\begin{figure}[ht!]
    \centering
    \includegraphics[scale=0.5]{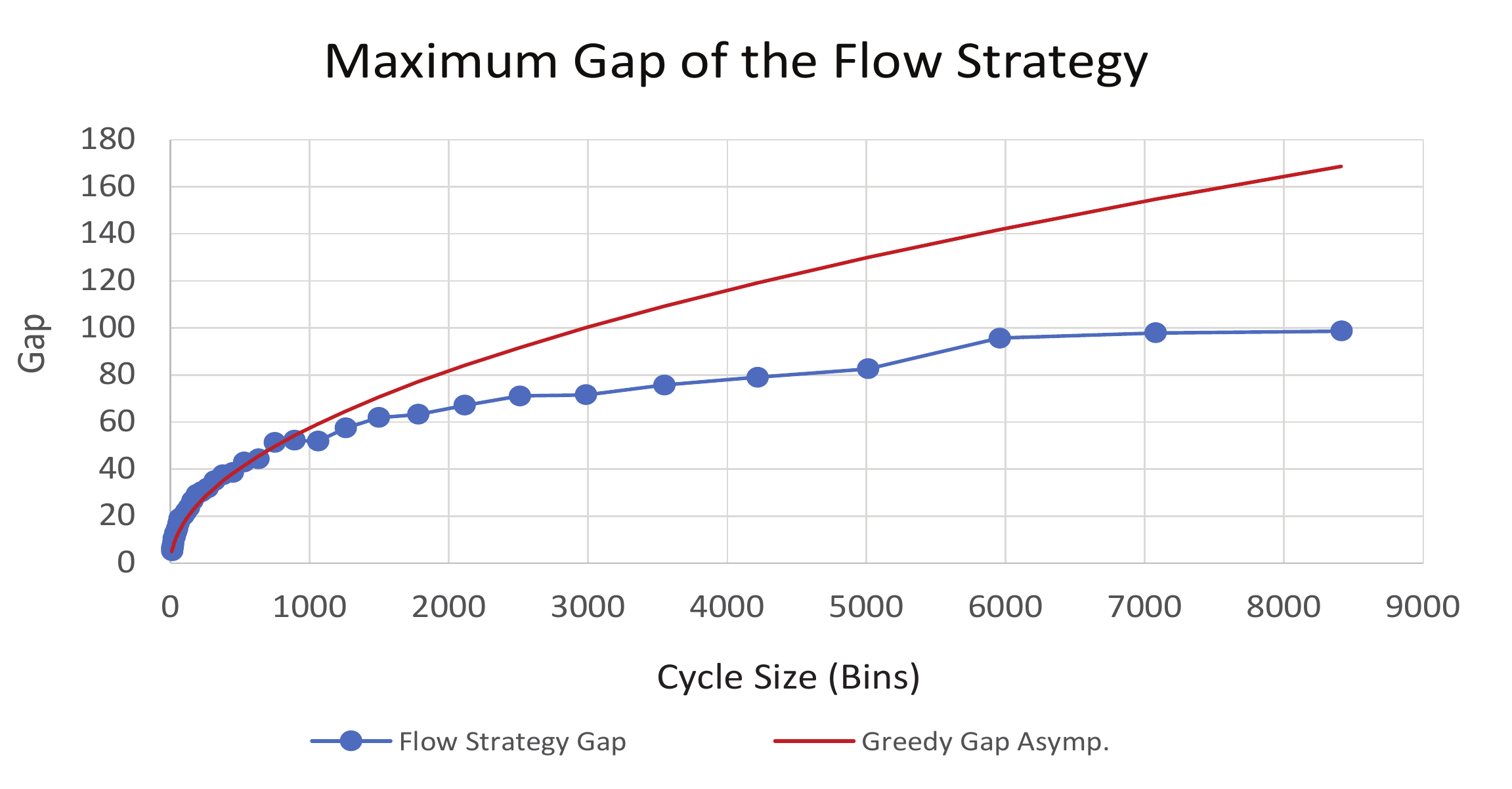}
    \caption{The gap for our flow algorithm on cycles for the best possible R\"{a}cke tree, averaged over 32 runs of $n^{2.5}$ balls for cycles of sizes $n=10$ to $8400$.   
    }
    \label{fig:Flow Gap}
\end{figure}

\section*{Acknowledgments}
We  thank  Guy Bensky and Shlomo Ron for their coding assistance in making our simulations.
\small{
\bibliographystyle{alpha}
\bibliography{ref}
}
\end{document}\documentclass{article}[11pt]

\usepackage[fencedCode,inlineFootnotes,citations,definitionLists,hashEnumerators,smartEllipses,hybrid]{markdown}
\usepackage[utf8]{inputenc}
\usepackage[T1]{fontenc}
\usepackage[scr=txupr]{mathalfa}
\usepackage{xcolor}
\usepackage[british]{babel}
\usepackage[autostyle]{csquotes}
\usepackage{mathtools}
\usepackage{amsmath, amsthm,amssymb,amstext,amsbsy} 
\usepackage{tabularx}
\usepackage{thmtools,thm-restate}
\usepackage{fullpage}
\usepackage{todonotes}[color=white]
\usepackage{enumitem}
\usepackage{caption}
\captionsetup{width=1.2\linewidth} 
\setlist[description]{leftmargin=15pt,labelindent=15pt}

\usepackage[bookmarks=false]{hyperref}  
\definecolor{darkgreen}{rgb}{0,0.5,0}
\hypersetup{linktocpage=true,colorlinks=true,citecolor=red,linkcolor=darkgreen}

\makeatletter
 \addtolength{\partopsep}{-2mm}
 \setlength{\parskip}{5pt plus 1pt}
 \addtolength{\abovedisplayskip}{-3mm}
 \addtolength{\textheight}{35pt}
\makeatother
\usepackage{setspace}



\newif\ifnotes\notestrue

\ifnotes
 \usepackage{color}
 \definecolor{mygrey}{gray}{0.50}
 \newcommand{\notename}[2]{{\textcolor{red}{\footnotesize{\bf (#1:} {#2}{\bf ) }}}}
 \newcommand{\noteswarning}{{\begin{center} {\Large WARNING: NOTES ON}\end{center}}}
 
\else
 
 \newcommand{\notename}[2]{{}}
 \newcommand{\noteswarning}{{}}
\fi

\newcommand{\nnote}[1]{{\color{red}\notename{Nikhil}{#1}}}
\newcommand{\onote}[1]{{\color{blue}\notename{Ohad}{#1}}}

\newtheorem{theorem}{Theorem}[section]
\newtheorem{claim}[theorem]{Claim}
\newtheorem{propos}[theorem]{Proposition}
\newtheorem{lemma}[theorem]{Lemma}
\newtheorem{corollary}[theorem]{Corollary}
\newtheorem{conjecture}[theorem]{Conjecture}
\newtheorem{remark}[theorem]{Remark}
\newtheorem{definition}[theorem]{Definition}
\newtheorem{obs}{Observation}
\newtheorem{thm}{Theorem}[section]

\newcommand{\expref}[2]{{\texorpdfstring{\hyperref[#2]{#1~\ref{#2}}}{#1~\ref{#2}}}} 
\newcommand{\secref}[1]{\expref{Section}{#1}}
\newcommand{\thmref}[1]{\expref{Theorem}{#1}}
\newcommand{\clmref}[1]{\expref{Claim}{#1}}
\newcommand{\tref}[1]{\expref{Theorem}{#1}}
\newcommand{\appref}[1]{\expref{Appendix}{#1}}
\newcommand{\lref}[1]{\expref{Lemma}{#1}}
\newcommand{\corref}[1]{\expref{Corollary}{#1}}
\newcommand{\conjref}[1]{\expref{Conjecture}{#1}}
\newcommand{\pref}[1]{\expref{Proposition}{#1}}
\newcommand{\figref}[1]{\expref{Figure}{#1}}

\newcommand{\boldp}{\ensuremath{\mathsf{p}}}

\newcommand{\ch}{\mathcal{C}}
\newcommand\defeq{\stackrel{\mathclap{\small\mbox{def}}}{=}}
\newcommand{\RZ}{\mathbb{Z}}
\newcommand{\R}{\mathbb{R}}
\newcommand{\E}{\mathbb{E}}
\newcommand{\N}{\mathbb{N}}
\newcommand{\FS}{\mathfrak{S}}
\newcommand{\CV}{\mathcal{V}}
\newcommand{\CD}{\mathcal{D}}
\newcommand{\ip}[2]{\langle #1,#2\rangle}
\newcommand{\sN}{\mathsf{N}}
\newcommand{\BR}{\mathbb{R}}
\newcommand{\BE}{\mathbb{E}}
\newcommand{\BP}{\mathbb{P}}
\newcommand{\CP}{\mathfrak{P}}
\newcommand{\bits}{\{0,1\}}
\newcommand{\pmone}{\{\pm1\}}
\newcommand{\ind}{\mathbf{1}}
\newcommand{\RN}{\mathbb{N}}
\newcommand{\Psijk}[1]{\Psi_{\mathbf{#1}}}
\newcommand{\CH}{\mathcal{H}}
\newcommand{\calP}{\ensuremath{\mathcal{P}}}
\newcommand{\CI}{\mathcal{I}}
\newcommand{\CE}{\mathcal{E}}
\newcommand{\bx}{\mathbf{x}}
\newcommand*\bh{\ensuremath{\boldsymbol{h}}}
\newcommand*\bj{\ensuremath{\boldsymbol{j}}}
\newcommand*\bk{\ensuremath{\boldsymbol{k}}}
\newcommand*\bl{\ensuremath{\boldsymbol\ell}}
\newcommand*\bm{\ensuremath{\boldsymbol{m}}}
\newcommand{\fhat}{\widehat{f}}
\newcommand{\dhat}{\widehat{d}}
\newcommand{\polylog}{\mathrm{polylog}}
\newcommand{\poly}{\mathrm{poly}}
\newcommand{\infnorm}[1]{\left\| #1 \right\|_{\infty}}
\newcommand{\eps}{\epsilon}
\newcommand{\disc}{\mathsf{disc}}
\newcommand{\sign}{\mathsf{sign}}
\newcommand{\CB}{\mathcal{B}}
\newcommand{\CL}{\mathcal{L}}
\newcommand{\CS}{\mathcal{S}}
\newcommand{\CN}{\mathcal{N}}
\newcommand{\CZ}{\mathcal{Z}}
\newcommand{\sfp}{\mathsf{p}}
\newcommand{\comp}[1]{\overline{#1}}
\newcommand{\err}{\mathsf{err}}
\newcommand{\CG}{\mathcal G}
\newcommand{\CT}{T}
\newcommand{\calT}{\mathcal{T}}
\newcommand{\smin}{\eps_\mathsf{min}}
\newcommand{\smax}{\eps_\mathsf{max}}
\newcommand{\cov}{\mathbf{\Sigma}}
\newcommand{\error}{\mathsf{err}}
\global\long\def\norm#1{\left\Vert #1\right\Vert }
\newcommand{\leaf}{l}
\newcommand{\bK}{\overline{K}}
\newcommand{\BS}{\mathbb{S}}
\newcommand{\p}{\mathbb{P}}
\newcommand{\sfq}{\mathsf{q}}
\newcommand{\sfu}{\mathsf{u}}
\newcommand{\tM}{M^+}
\newcommand{\op}{\mathsf{op}}
\newcommand{\diam}{\mathsf{diam}}
\newcommand{\im}{\mathrm{im}}
\newcommand{\SE}{\mathscr{S}}
\newcommand{\SA}{\mathscr{A}}
\newcommand{\SB}{\mathscr{B}}
\newcommand{\SG}{\mathscr{G}}
\newcommand{\ST}{\mathscr{T}}
\newcommand{\Tr}{\mathrm{Tr}}
\newcommand{\SZ}{\mathscr{Z}}
\newcommand{\SD}{\mathscr{D}}
\newcommand{\SL}{\mathscr{L}}
\newcommand{\arbnorm}[1]{\left\|#1\right\|_*}
\newcommand{\qlt}{q^l_{j,k}}
\newcommand{\qrt}{q^r_{j,k}}
\newcommand{\dlt}{d^l_{j,k}}
\newcommand{\drt}{d^r_{j,k}}
\newcommand{\SM}{\mathscr{M}}
\newcommand{\SN}{\mathscr{N}}
\newcommand{\ovl}{\overline{L}}
\newcommand{\putat}[3]{\begin{picture}(0,0)(0,0)\put(#1,#2){#3}\end{picture}}

\DeclareMathOperator{\nl}{nl}
\DeclareMathOperator{\bin}{bin}
\DeclareMathOperator{\sgn}{sgn}
\DeclareMathOperator{\total}{Total}
\DeclareMathOperator{\child}{child}
\DeclareMathOperator{\sib}{sib}
\DeclareMathOperator{\gap}{gap}
\DeclareMathOperator{\conge}{cong}

\let\oldabstract\abstract
\let\oldendabstract\endabstract
\makeatletter
\renewenvironment{abstract}
{\renewenvironment{quotation}%
               {\list{}{\addtolength{\leftmargin}{1em} 
                        \listparindent 1.5em%
                        \itemindent    \listparindent%
                        \rightmargin   \leftmargin%
                        \parsep        \z@ \@plus\p@}%
                \item\relax}%
               {\endlist}%
\oldabstract}
{\oldendabstract}
\makeatother


\title{Well-Balanced Allocation on General Graphs}
\author{Nikhil Bansal\thanks{CWI Amsterdam and TU Eindhoven, \texttt{N.Bansal@cwi.nl}. Supported by the  NWO VICI grant 639.023.812.} \and {Ohad Feldheim\thanks{Hebrew University of Jerusalem Israel, \texttt{ohad.feldheim@mail.huji.ac.il}. Supported by ISF grant 1327/19.}}
}

\date{}

\begin{document}

\maketitle

\begin{abstract}

 We study the graphical generalization of the 2-choice balls-into-bins process, where rather than choosing any two random bins, the bins correspond to vertices of an underlying graph, and only the bins connected by an edge can be chosen.
 
    For any $k(n)$ edge-connected, $d(n)$-regular graph on $n$ vertices and any number of balls, we give an allocation strategy 
 which guarantees that the maximum gap between the bin loads is $O((d/k) \log^4\hspace{-1pt}n \log \log n)$, with high probability. We further show that the dependence on $k$ is tight and give an $\Omega((d/k) + \log n)$ lower bound on the gap achievable by any allocation strategy, for any graph $G$. In particular, our result gives polylogarithmic bounds for natural graphs such as cycles and tori, where the classical greedy allocation appears to result in a polynomial gap. Previously such a bound was known only for graphs with good expansion.

    The construction is based on defining certain orthogonal flows on cut-based R\"{a}cke decomposition of graphs. The allocation algorithm itself, however, is simple to implement and takes only $O(\log(n))$ time per allocation, and can be viewed as a global version of the greedy strategy that compares average load on sets of vertices, rather than on individual vertices.
    
    \smallskip
\noindent \textbf{Keywords.} Load-balancing, Balls-into-bins processes, graphical two-choice, R\'{a}cke decomposition.
    
\end{abstract}

\section{Introduction} 

Randomized balls-into-bins models serve as useful abstractions for various problems arising in hashing, load balancing and resource allocation in parallel and distributed systems and have been extensively studied in the areas of probability, economics and algorithms (see e.g.,~\cite{RS98,DR96,AK14}).
The balls typically represent tasks or items, that need to be allocated to resources that are modeled by the bins, and the goal is to minimize either the maximum load or the gap between the maximum and minimum load as much as possible.

The models differ by what kind of control is available to the algorithm over the allocation process.
In the classical {\em single-choice} model, the algorithm has no control and each ball is placed in a bin chosen uniformly at random. For $m$ balls and $n$ bins, it is well known that
for $m=n$, the heaviest bin has load $(1+o(1)) \ln n/\ln \ln n$ with high probability (w.h.p.).
For $m\geq n \log n$, the bins have load in the range  $m/n \pm \Theta(\sqrt{(m \log n)/n})$ i.e.,~deviating from the average load of $m/n$ by $\Theta(\sqrt{(m \log n)/n})$. 

Perhaps the simplest and most well-studied controlled variant of this model is the \emph{2-choice} model. In this model, at each step the algorithm is given two uniformly chosen bins into one of which it must allocate the ball. This modification, which may appear minor, leads to 
substantial improvements. In a seminal result, Azar, Broder, Karlin and Upfal \cite{ABKU94} showed that if a ball is placed in the least loaded of $d\geq 2$ uniformly sampled bins, then for $m=O(n)$, the maximum load reduces to $\ln \ln n/\ln d + \Theta(m/n)$. They also establish that, asymptotically, this {\em greedy} allocation strategy is optimal for this model.
These results were extended by Berenbrink, Czumaj, Steger and V{\"o}cking~\cite{BCSV06} to arbitrary $m$, who showed that the maximum load is
$m/n + \ln \ln n/\ln d + O(1)$ with probability $1-1/\text{poly}(n)$. It is worth noting that even for $d=2$ choices, the excess load over the average does not increase with $m$, unlike for the case of $d=1$.

\paragraph{Graphical process.}
In many natural settings, there are restrictions on which pairs of bins can be queried or where the ball can be placed. An elegant generalization of the $2$-choice process, called the {\em graphical process}, was introduced by Kenthapadi and Panigrahy \cite{KP06}.
Here there is an underlying graph $G=(V,E)$ on $n=|V|$ vertices, and 
at each step, a uniformly random edge $e=(u,v)$ is chosen and the ball must be placed on one of the two endpoints of $e$. Notice that the $2$-choice process corresponds to $G=K_n$.
This motivates the following natural question:

{\em Given a graph $G$ what is the best load balance obtainable by a graphical two-choice allocation strategy?}

Extending the results for the classical $2$-choice process ($G=K_n$), 
\cite{KP06} showed that if $G$ is $n^{\epsilon}$-regular, then for $m=n$ balls, the greedy strategy, has  maximum load is $\ln \ln n + O(\log 1/\eps)$. An extension to hypergraphs was considered in \cite{Godfrey08}. 

The setting of arbitrary $m$, also the focus of our work, was considered by Peres, Talwar, Weider \cite{PTW15}. They investigated the greedy strategy,
and showed a gap of $\Theta(\log n)$ for regular expander graphs between the maximum and minimum bin load. 
This gap is the best possible for any strategy, up to constants.
More generally they showed that the gap is $O((\log n)/\beta)$ 
for any $d$-regular graph with edge-expansion\footnote{For any subset $S \subset V$ with $|S|\leq n/2$, $E(S,\overline{S}) \geq \beta d |S|$.} $\beta$.

\begin{remark}
Let us define  {\em upper gap} as the difference between the maximum load and the average load, and {\em gap} as the maximum difference between the bins loads. 
Both of these objectives have been studied extensively, and they can sometimes differ significantly, e.g.,~for $G=K_n$ where the gap is $\Theta(\log n)$ and the upper gap is $\Theta(\log \log n)$. However, this difference often disappears in the setting of general graphs, or more general choice models that we consider here. For example, for a constant degree expander $G$, the gap is $\Theta(\log n)$ while the upper gap is $\Omega(\log n/\log \log n)$. We discuss this further in Sections \ref{sec:related} and \ref{sec:lb}.
\end{remark}

\paragraph{Graphs with low expansion and limitations of the greedy strategy.} While the upper-bound of $O((\log n)/\beta)$ by Peres et al.,~implies a poly-logarithmic gap for well-expanding graphs, it gives only polynomially large $O(n \log n)$ and $O(n^{1/2} \log n)$ bounds for graphs such as cycles or two-dimensional grids, with low expansion.
The dependence on $\beta$ in the result of \cite{PTW15} is not tight
and they leave open the problem of getting the right bound, even for these simple classes of graphs.

While for certain non-expander graphs, such as high-dimensional balanced grids, it is expected that the greedy algorithm will result in  polylogarithmic gaps, and only the tools to prove it are missing, there are many graphs on which the greedy strategy appears to be inherently quite sub-optimal. An instructive example is that of a cycle, which has also been studied on its own. 
Here, the conjectured gap and upper gap estimates for the greedy strategy are $\Omega(\sqrt{n})$. Some evidence for this conjecture was given by Alistarh, Nadiradze and Sabour in \cite{ANS20}. They consider a model in which after allocating the ball to one of the vertices of the requested edge, the new load of the two vertices becomes the average of their load. Even in this model, which intuitively should have a lower gap than the standard greedy algorithm, they show an $\Omega(\sqrt{n})$ lower bound for the typical gap. We also provide supporting simulation results for this prediction in figure~\ref{fig:Greedy Gap}.

\begin{figure}[h!]
    \centering
    \includegraphics[scale=0.5]{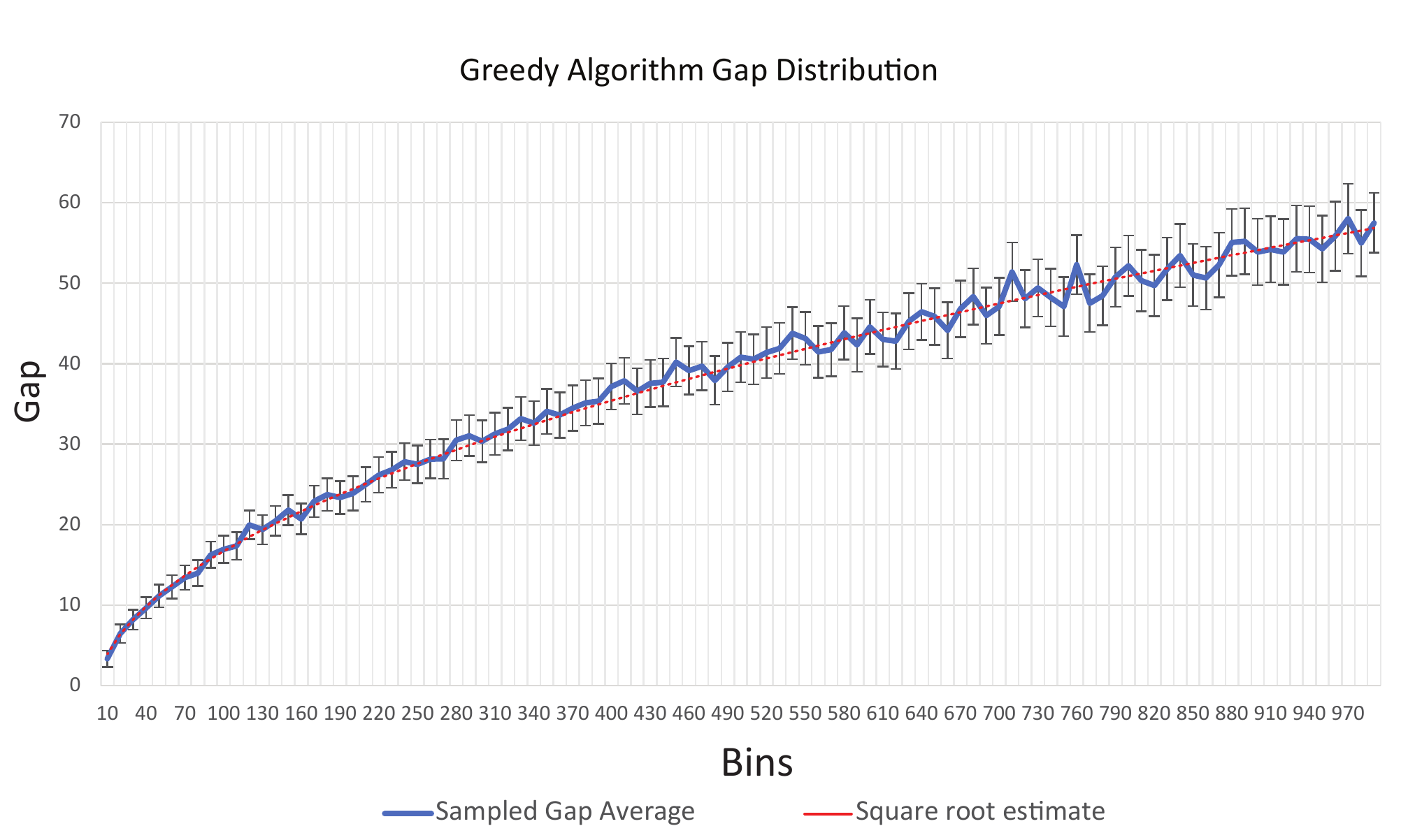}
    \caption{The gap for the greedy algorithm on cycles, averaged over 84 runs of $10^9$ balls for cycles of sizes from 10 to 1000. The dotted guide is the function $f(x)=1.85\sqrt{x}-1$, error margins for 95\% confidence are provided. The graph clearly shows the polynomial growth of the gap.
    }
    \label{fig:Greedy Gap}
\end{figure}

More generally, it has been conjectured that the load fluctuations under greedy 
are similar to the fluctuations in classical statistical mechanics models. In particular Peres (in private communication) suggested that, up to a $\log n$ factor the gap should be the same as that of a Gaussian free field on $G$. Both these conjectures and simulation results suggest that on very poorly expanding graphs, such as cycles and unbalanced tori the imbalance in load should be polynomial in $n$.

This raises the natural question whether there exist other allocation strategies, ideally simple to implement, which achieve substantially better gaps on such graphs.

\subsection{Our Results}
\label{sec:results}
We give an allocation strategy that achieves the best possible gap, up to polylogarithmic factors, for the graphical process on any graph $G$. 
More formally, we show the following.
\begin{thm}\label{thm:main}
Let $G=(V,E)$ be any $k$ edge-connected, $d$-regular graph on $n$ vertices. There is an allocation strategy for the graphical process on $G$, that guarantees for any time $t\in\N$,  \[\gap_G(t) = O((d/k) \log^4 n \log \log n)\]
with probability at least $1-1/\text{poly}(n)$. Here $\gap_G(t)$ is the maximum difference in vertex loads at time $t$. 
\end{thm}
In addition, we show that this is the best possible, up to polylogarithmic factors, for every graph $G$. 
\begin{thm}
\label{thm:lb}
For any $d$-regular graph $G$ and for any allocation strategy for $G$, for any time $t$, with constant probability, $gap_G(t) = \Omega(d/k + \log n)$, where $k$ is the edge-connectivity of $G$.
\end{thm}

Theorem \ref{thm:main} implies an allocation strategy with gap $\text{polylog}(n)$ for cycles and grids, and more generally for any  graph $G$ with $k = \Omega(d/\text{polylog}(n))$.
 
In fact, our strategy is shown to have gap $O(\alpha_G (d/k) \log^2 n )$ where $\alpha_G$ is the congestion ratio for oblivious routing on $G$ based on a R\"{a}cke decomposition tree \cite{R02} (see Section \ref{sec:prel}).
The bound in Theorem \ref{thm:main} follows from a result of Harrelson, Hildrum and Rao~\cite{HHR03} who showed that  $\alpha_G = O(\log^2 n \log \log n)$ for any $G$.  

\paragraph{The Algorithm.}
 The allocation strategy is simple and incurs only $O(\log n)$ worst case running time per allocation. It can be viewed as a more {\em global} version of the greedy strategy, where instead of comparing the loads on two vertices, one
compares the average load on two random sets chosen from a fixed small collection.

More specifically, the algorithm maintains a table of size $O(n)$ with entries consisting of the average load on certain subsets $S$ of $V$.
When a random edge $e=(x,y)$ is requested, it picks two sets $S,S'$ according to a {\em fixed} distribution $P_e$ specifying probabilities $p_e(S,S')$ over pairs of these subsets. It then assigns the ball to $x$ or $y$ depending on which among $S$ and $S'$ has larger average load.
The sets $S$ form a balanced hierarchical decomposition of $V$, so when a ball is assigned to some vertex $v$, only the $O(\log n)$ entries for sets $S$ containing $v$ are updated.

\paragraph{Ideas and Techniques.} 
The main tools in obtaining Theorem \ref{thm:main} are \emph{binary tree decompositions} and \emph{orthogonal flows}. This method, inspired by discrepancy theory tools, defines pairs of sets on the graph in a hierarchical fashion 
and constructs the distributions $P_e$ using a mixture of probabilistic allocation strategies. Each of these strategies affects only the relative allocation probabilities of the sets in a pair, without affecting these for other pairs. For a more thorough overview of our method see Section~\ref{sec:overview}.

The lower bound in Theorem \ref{thm:lb} is based on a simple observation is that if we fix any minimum cut $(S,\overline{S})$ of size $k$, there are not enough edges crossing the cut to balance out the random load fluctuations that arise due to the requests for edges with both endpoints  in $S$.

Finally, we make a few remarks concerning the model and our results.

\begin{remark} The requirement that $G$ is regular is standard in the area, and is made so that the expected load under a random strategy is equal on all bins. The results extend to irregular graphs in a natural way where one measures the gap with loads normalized suitably by the degree. 
\end{remark}
\begin{remark} Our results do not require the full power of two choices. In the $1+\beta$ graphical choice model, at each step an edge is given only with probability $\beta$, and with probability $1-\beta$ the algorithm is given no choice and the ball is allocated to a random bin. Our bound on the gap extends directly to this model with factor $O(1/\beta)$ loss. More background on this model is provided in Section~\ref{sec:related} below.
\end{remark}
\begin{remark}
While better bounds are known for congestion in oblivious routing  \cite{AzarCFKR04, Racke08, AndersenFeige}, our method requires a single 
R\"acke decomposition tree as the demands we define for our flows depend on the tree itself.
\end{remark}

\subsection{Related work and models}
\label{sec:related}
The literature on ball-into-bins processes is extensive and 
it is impossible to cover even a small portion of the developments and applications.
The first appearance of a 2-choice type result was in \cite{KLM96} in the context of online hashing. Following the result of Azar et al.,~\cite{ABKU94} several variations of the model have been studied.
 Variations include models in which balls are eliminated over time \cite{Mthesis91,CFMMSU98} -- either by age or at random and parallel allocation of the balls with limited communication \cite{Stemann96,ACM98}. Many of the earlier results are surveyed in Mitzenmacher's Thesis \cite{Mthesis91} 
and in 
his survey with Richa and Sitaraman \cite{MRS01}. A more recent survey is due to Wieder \cite{Wieder}.

In his thesis  Mitzenmacher suggested the model of $1+\beta$ choice for $\beta<1$, where the algorithm is given two choices with probability $\beta$ and only one choice with probability $1-\beta$. His motivation for introducing this 
model stems from a problem in queuing theory. V\"{o}cking \cite{Vocking03} show that for $d$-choice, non-uniform choices can improve over the greedy algorithm of \cite{ABKU94}, resulting in maximum load of $\Theta(\log\log (n)/d)$ (cf.,~$\Theta(\log\log (n)/\log d)$). 

The heavily loaded case of the two-choice problem was first analyzed in  \cite{BCSV06} (see a neat and short proof by Talwar and Wieder \cite{TW14}). In \cite{PTW15}, Peres, Talwar and Wieder considered the $1+\beta$ choice model for complete graphs and showed that there both gap and upper gap are $\Theta((\log n)/\beta)$. The drift and potential methods introduced in this work 
allowed the authors to relate this result to the graphical case for expanders in \cite{PTW15} and inspired methods used here as well.

\subsection{Notation and Preliminaries}\label{sec:prel}
A graphical process is specified by some fixed underlying $d$-regular graph $G=(V,E)$. Henceforth we denote $n=|V|$ and $m=|E|$ (typically in balls-into-bins literature, $m$ is used for the number of balls, but we will think of the allocation process as indefinite, using $t$ for the index of the allocated ball).
At each time step $t=1,2,\ldots,$ an  edge $e=e_t=(u,v) \in E$ is chosen (\emph{requested}) uniformly at random, and a ball must be assigned to one of its endpoints $u$ or $v$. We often refer to the vertices of $G$ as bins.
A vertex $u$ has load $\ell$ at time $t$, if $\ell$ balls have been assigned to it after $t$ allocations.
Let $L^t:V \rightarrow \N$ denote the load vector at time $t$. 
We assume $L^0(u)=0$ for all $u \in V$ so that the total load $\|L^t\|_1 =t$ for each $t\in \N$. Hence, the average load at a vertex at time $t$ is $t/n$. 

An allocation strategy, upon the request $e=(u,v)$, decides whether to assign the ball to $u$ or to $v$ (possibly based on the entire history so far). The goal of the strategy is to minimize the gap, where we define the gap at time $t$ as 
\[ \gap(t)=\gap_G(t) = \max_u L^t(u)  - \min_u L^t(u).\]
As $ \max_u | L^t(u) - t/n| \leq \gap_G(t) \leq  2 \max_u |L^t(u)-t/n|$, sometimes we will work with $\max_u | L^t(u) - t/n|$ and we refer to this as the  maximum deviation from the average load.
For a subset $S \subset V$, we denote the total load on vertices in $S$ by $L^t(S) := \sum_{u \in S} L^t(u)$, and the average load by $\overline{L}^t(S) := L^t(S)/|S|$. 

\paragraph{Hierarchical cut-based decomposition.}
A {\em hierarchical decomposition} of $G$ is a recursive partition of the vertex set $V$ until each resulting set is a singleton vertex.
Such a decomposition is naturally viewed as a rooted tree $T=(V_T,E_T)$, where each node $i\in V_T$ corresponds to a subset $S_i \subset V $. The root $r$ of $T$ corresponds to $V$, the leaves to singleton sets $\{u\}$ for $u\in V$. 
For any node $i \in V_T$, its children $j_1,\ldots,j_k$ correspond 
to the sets $S_{j_1},\ldots,S_{j_k}$ obtained by partitioning $S_i$.

Removing an edge $(i,j)$ of $T$, where $j$ is a child of $i$, partitions the leaves of $T$ (or nodes of $G$) into $S_j$ and $V\setminus S_j$. In a {\em cut-based} decomposition, we associate an edge $(i,j)\in E_T$ with the cut $(S_j,V\setminus S_j)$ of $G$, and set its capacity 
$c_T(i,j) = C_G(S_j)$, the capacity of the cut $(S_j, V \setminus S_j)$ in $G$.

We will work with $G$ and its decomposition $T$. To avoid confusion, we use $u, v,x,y$ to index the vertices of $G$, and $i,j,j'$ to index the vertices of $T$. 
Sometimes we index the leaves of $T$ by $u,v$ to highlight the correspondence with vertices of $G$. 
We will refer to the vertices of $T$ as nodes.

\paragraph{
R\"{a}cke Trees and Oblivious Routing.}
Let $G=(V,E,c)$ be an undirected graph with edge capacities $c_e \geq 0$. 
An oblivious routing for $G$ specifies for each ordered pair of nodes $u,v \in V$, a  flow template $f_{uv}: V \times V \rightarrow \R$ on the edges, that describes how to send one unit of flow from $u$ to $v$. So each $f_{uv}$ satisfies
\[ f_{uv}(x,y) = -f_{uv}(y,x), \text{ for all } x,y \in V, \text{ and } \sum_y  f_{uv}(x,y) = {1}_{(x=u)} - 1_{(x=v)}   \text{  for all } x\in V.\] 
Let  $\vec{d} = (d(u,v))_{u,v}$ be any multicommodity demand vector with commodities $(u,v)\in V\times V$ and demands $d(u,v)\geq 0$. An oblivious routing of $\vec{d}$ on $G$ sends $d(u,v)$ units of each commodity $(u,v)$ according to $f_{uv}$.

In a breakthrough work, R\"{a}cke \cite{R02} showed that for any edge capacitated $G=(V,E,c)$, there exists a cut-based decomposition tree $R=(V_R,E_R,c_R)$, and flow templates $f_{uv}$ for each $u,v \in V$ satisfying the following remarkable property.
Consider any demand vector $\vec{d}$.
On $R$, as the $d(u,v)$ units are routed along the unique path from leaves $u$ to $v$, let  $g_{R,\vec{d}\hspace{1pt}}(i,j)=
\sum_{u \in S_j, v \notin S_j} d(u,v) + \sum_{u \notin S_j, v \in S_j}$, be the total flow across $(i,j) \in E_R$ and let 
$\conge(R,\vec{d}) = \max_{(i,j) \in E_R} g_{R,\vec{d}\hspace{1pt}}(i,j)/c_R(i,j)$
be the maximum edge-congestion on $R$.
Similarly, consider the oblivious $g_{G,\vec{d}}$ for $\vec{d}$ on $G$ and let  $\conge(G,\vec{d}) = \max_{e \in E} g_{G,\vec{d}\hspace{1pt}}(e)/c_e $ be the maximum edge-congestion on $G$.
Then for any $\vec{d}$,
\[ \conge(R,\vec{d}) \leq  \conge(G,\vec{d}) \leq \alpha_G\, \conge(R,\vec{d}).\]
We refer to $\alpha_G$ as the congestion ratio for oblivious routing on $G$. The best bound on $\alpha_G$ is $O(\log^2 n \log \log n)$ \cite{HHR03}, and  the corresponding $R$ and flow templates $f_{uv}$ can be found in polynomial time. Almost linear time constructions with slightly worse $\alpha_G$ are also known \cite{RST14}.
Below, we also use the fact that in both of these constructions, the tree $R$ satisfies $|S_j| \leq 3/4 |S_i|$ for each edge $(i,j)$ and that $R$ has depth $O(\log n)$.

\section{Overview}\label{overview}
Before describing the details, we first provide a high level overview of the idea and the algorithm.

Upon the arrival of an edge request $e=(u,v)$, by choosing whether to allocate the ball to $u$ or $v$, an allocation strategy can {\em bias} the expected load toward $u$ or $v$. 
As the vertex loads fluctuate below and above the average load over time, these biases must depend on the current load vector, and the goal is to design an allocation strategy that results in a self-regulating process that keeps the load deviations small.

\paragraph{Balancing average load on sets.} Let $G=(V,E)$ be the graph underlying the process.
Consider some {\em binary} hierarchical decomposition of $G$, represented by a binary tree $T = (V_T,E_T)$, with nodes $i \in V_T$ corresponding to subsets $S_i \in V$ and leaves --- to vertices in $G$.
Let $r$ be the root of $T$, and for a leaf $u$, consider the unique path $u=i_0, i_1, \ldots, i_h=r$ from $u$ to $r$ in $T$, so that $\{u\}=S_{i_0} \subset S_{i_1} \subset \ldots \subset S_{i_h}= V$. 
As $\overline{L}^t(r) = t/n$ and $\overline{L}^t(u)=L^t(u)$, the load deviation for $u$ at time $t$ is at most \[|L^t(u)-t/n |= \left|\sum_{j=1}^h \big(\overline{L}^t(S_{i_{j-1}})- \overline{L}^t(S_{i_j})\big) \right|\leq    \sum_{j=1}^h \left|\overline{L}^t(S_{i_{j-1}})- \overline{L}^t(S_{i_j})\right|,\]
and hence to control the load deviation for each $u \in V$, up to an $h=O(\log n)$ factor, it suffices to control the gap $|\overline{L}^t(S_{i_{j-1}})- \overline{L}^t(S_{i_j})|$ for each parent-child node pair.

Fix a non-leaf node $i\in V_T$, and let $\ell(i), r(i)$ denote its left and right children. To bound the average load gap between $S_i$ and its children, 
a simple computation shows that it suffices to bound $|\overline{L}(S_{\ell(i)}) - \overline{L}(S_{r(i)})|$. So a natural idea is to assign each non-leaf node $i$ the task of balancing $\overline{L}(S_{\ell(i)})$ and $\overline{L}(S_{\ell(i)})$. To do this, one can try to assign biases to edges to create a \emph{relative bias} from the most loaded among $S_{\ell(i)}$ and $S_{r(i)}$ to the least loaded.
Also the relative biases should be strong enough for each sibling pair, to get good bounds on the balance. 
In particular, if $k=\Omega(d)$, then to achieve the desired polylogarithmic gap in Theorem \ref{thm:main}, the total bias between every two siblings 
$S_{\ell(i)}$ and $S_{r(i)}$ must be $\Omega(d/\text{polylog}(n))$.

\begin{figure}[ht!]
    \centering
    \includegraphics[scale=0.38]{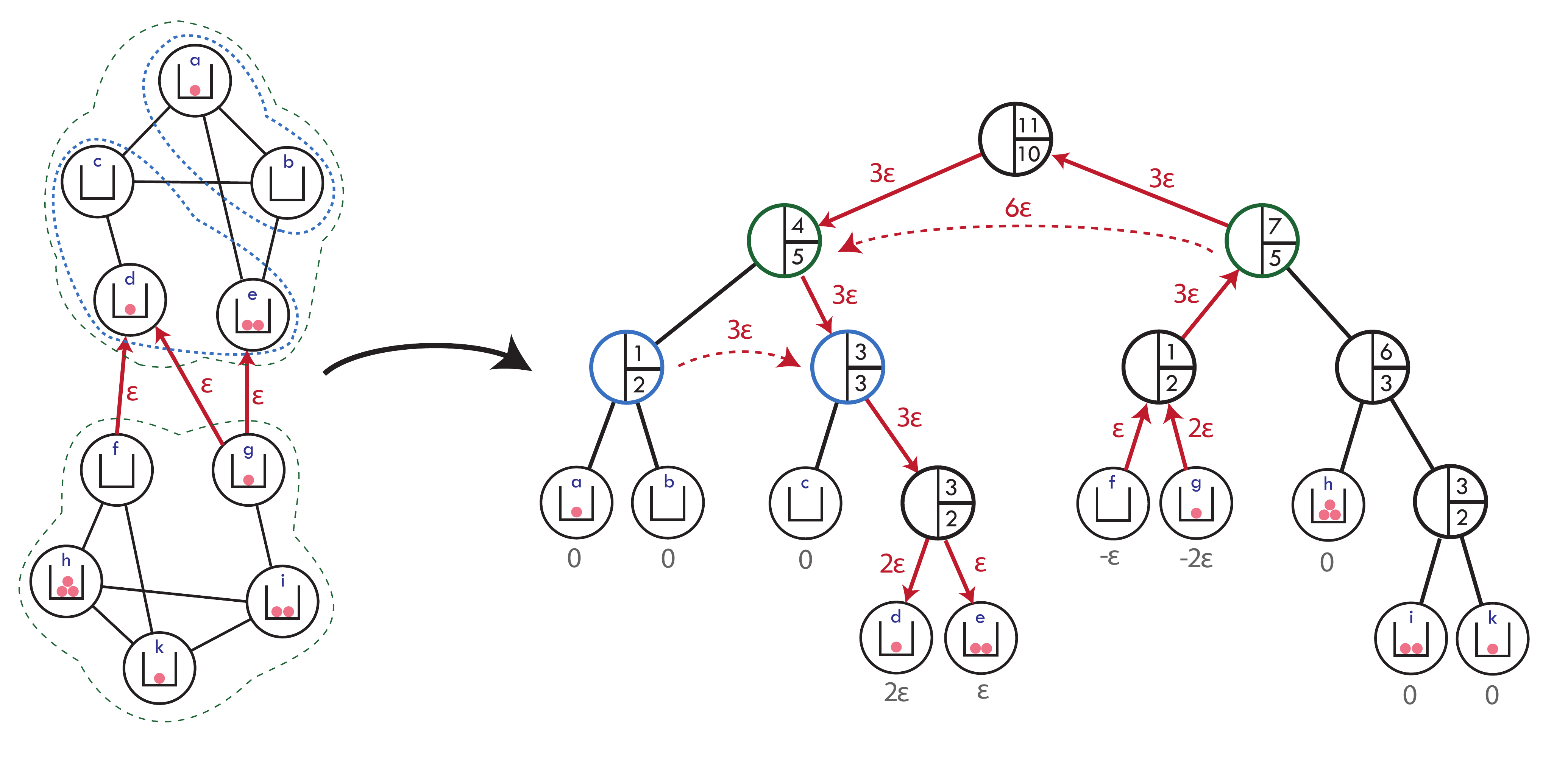}\\
    \putat{-140}{29}{\textcolor{darkgreen}{$S_2$}}
    \putat{-135}{208}{\textcolor{darkgreen}{$S_1$}}
    \putat{-3}{161}{\textcolor{darkgreen}{\scalebox{0.7}{$S_1$}}}
    \putat{131}{161}{\textcolor{darkgreen}{\scalebox{0.7}{$S_2$}}}
    \putat{-167}{188}{\textcolor{orange}{\scalebox{0.7}{$S_{11}$}}}
    \putat{-178}{159}{\textcolor{orange}{\scalebox{0.7}{$S_{12}$}}}
    \putat{-63}{124}{\textcolor{orange}{\scalebox{0.5}{$S_{11}$}}}
    \putat{-2}{124}{\textcolor{orange}{\scalebox{0.5}{$S_{12}$}}}\\[-17pt]
    \includegraphics[scale=0.38]{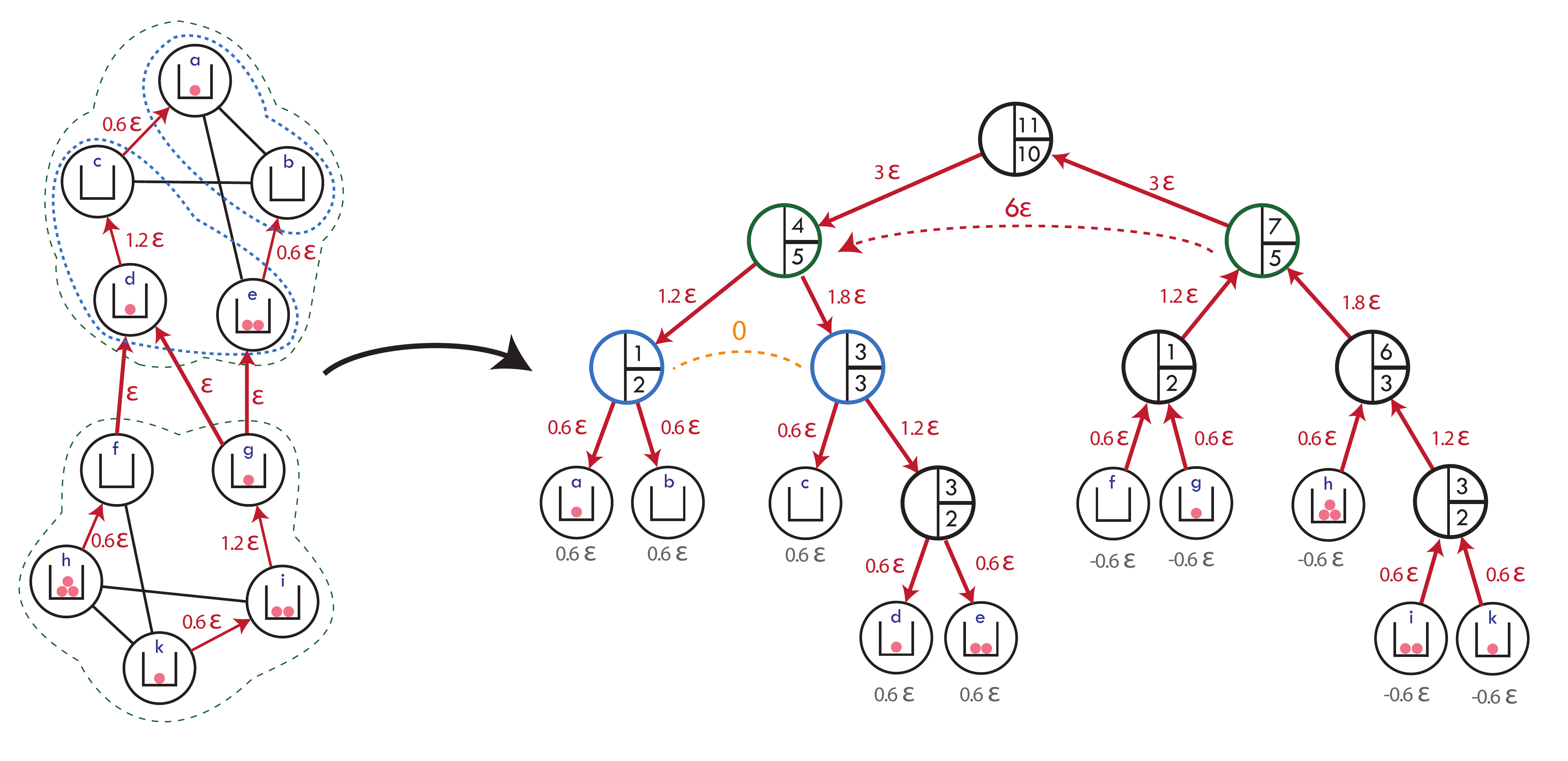}\\
    \putat{-140}{29}{\textcolor{darkgreen}{$S_2$}}
    \putat{-135}{208}{\textcolor{darkgreen}{$S_1$}}
    \putat{-3}{161}{\textcolor{darkgreen}{\scalebox{0.7}{$S_1$}}}
    \putat{131}{161}{\textcolor{darkgreen}{\scalebox{0.7}{$S_2$}}}
    \putat{-167}{188}{\textcolor{orange}{\scalebox{0.7}{$S_{11}$}}}
    \putat{-178}{159}{\textcolor{orange}{\scalebox{0.7}{$S_{12}$}}}
    \putat{-63}{124}{\textcolor{orange}{\scalebox{0.5}{$S_{11}$}}}
    \putat{-2}{124}{\textcolor{orange}{\scalebox{0.5}{$S_{12}$}}}\\[-10pt]  
    
    \caption{On the \textbf{left}, two different edge biases on a graph $G$ are illustrated by red arrows. Each of these is aimed at creating a relative bias from $S_2$ with average load $\overline{L}(S_2)=\tfrac75$ to $S_1$ with  $\overline{L}(S_2)=\tfrac45$. 
    On the \textbf{right} are their impact on the graph's decomposition tree. Each node of the tree represents a vertex subset $S$ in the graph. Inside the node is the average load for the set (balls divided by bins). The horizontal dashed edges indicate the relative bias between two siblings. The edge from a $S$ to its parent is oriented by the total bias on that $S$. \textbf{Above} we see the non-orthogonal flow incurred by the edges the cross the green cut. Observe how these create an undesired relative bias from the lightly loaded $S_{11}$ to the heavily loaded $S_{12}$.
    \textbf{Below} we see an orthogonal flow which has no impact on the relative bias of any pair of sibling sets except $S_1$ and $S_2$.}
    \label{fig:edge flow}
\end{figure}

\paragraph{Orthogonal Multicommodity Flows.}   A key difficulty in implementing this idea is that the bias on an edge $e$ affects several sets $S_i$, 
(see Figure \ref{fig:edge flow}).
In this example, we would like to use edges between $S_1$ and $S_2$ to create a bias from $S_2$, which has higher average load, to $S_1$. However, this create a bias towards $S_{12}$ compared to $S_{11}$, even though $S_{12}$ already has higher load.
In general, it is not clear how to create edge biases to balance the relative bias {\em simultaneously} for all sibling pairs.

To get around this, we do two things. First the graph decomposition is constructed carefully, as we discuss later. Second, for each non-leaf node $i$ in $T$, we create a flow $F_i$ between vertices in $S_{\ell(i)}$ and $S_{r(i)}$ 
We can view the flow across an edge $e$ as the bias that we would like to assign to $e$ to balance $\overline{L}(S_{\ell(i)})$ and $\overline{L}(S_{r(i)})$. So the direction of the flow $F_i$ depends on which side has higher average load.

However, as the bias on $e$ affects each set $S_j$ that $e$ crosses,
we choose the demands between vertices in $S_{\ell}(i)$ and $S_{r(i)}$ to define $F_i$ such that the flow $F_i$ does not affect the balancing task of any other non-leaf node $j \neq i$.
More formally for any non-leaf node $j \neq i$, the quantity $\overline{L}(S_{\ell(j)}) - \overline{L}(S_{\ell(j)})$ incurs zero bias due to the flow $F_i$. We refer to this as the {\em orthogonal} property of demands and it allows to create biases that simultaneously balance $\overline{L}(S_{\ell(j)}), \overline{L}(S_{r(j)})$ for all pairs. This also is illustrated in Figure~\ref{fig:edge flow}.

To do this, we define a multi-commodity flow for each tuple $(i,u,v)$ with $u \in S_{\ell(i)}$ and $v \in S_{r(i)}$ and assign it a demand $\kappa_i(u,v)$ proportional to $1/|S_{\ell(i)}||S_{r(i)}|$.  If any non-leaf node $j\neq i$ in the subtree rooted at $i$, this ensures that the flow $F_i$ is split in proportion to $|S_{\ell(j)}|$ and $|S_{r(j)}|$ along the two children of $j$, ensuring orthogonality as   $\overline{L}(S_{\ell(j)}) - \overline{L}(S_{r(j)})$ incurs zero bias. Moreover if $i$ and $j$ do not lie along some root-leaf path, the tree structure ensures that the flows $F_i$ and $F_j$ do not affect each other. 

In principle we could have generated these flows ad-hoc, however it is desirable that our allocation rule shall be polylogarithmically fast and simple to implement. The fact that we use a separate commodity for the flow between every pair of sibling sets allows us to simply change the sign of that flow when the load balance between the sets shifts. 

\paragraph{R\"acke Trees.} 
The above describes a strategy on a tree. To realize this strategy using flows, we use Raecke's cut-based decomposition for $G$. Roughly speaking, this ensures that the decomposition does not have bottlenecks and sufficiently high biases can be created between the relevant sets. 
Moreover in this realization the different commodities are \emph{oblivious} to each other, that is, the flow realization for each commodity depends only on the demand of that commodity. This enables the simple implementation of our algorithm.

In general, R\"acke trees need not be binary. To get around this we work with an auxiliary binary tree $T$ obtained from $R$, that we call a \emph{balancing-flow tree}. Even though $T$ is not a valid cut decomposition anymore,
it maintains the properties needed for our purpose.

\section{Algorithm}
Let $G=(V,E)$ be the underlying graph for the graphical process and let $k = k(G)$ denote the edge-connectivity (the size of the minimum cut) of  $G$.

The algorithm consists of two parts. A preprocessing stage, and an online allocation stage. In the preprocessing stage we analyze the graph structure and compute the decomposition tree. In the online allocation stage we use these to realize an allocation strategy that uses $O(\log n)$ operations per request to balance the load at all times.

\subsection{Computing the strategy}
Given $G$, we first  construct a R\"acke cut-based decomposition $R=(V_R,E_R,c_R)$ tree for $G$, together with the flow templates $f_{uv}$ for each  ordered pair $u,v \in V$, as defined in Section \ref{sec:prel}. 

Let $\alpha_G$ be the congestion ratio of $R$ with respect to $G$. We apply the algorithm of \cite{HHR03} to construct $R$, so that $\alpha_G =O(\log^2 n \log \log n)$ and the running time is polynomial.
Let $R_i \subset V$ denote the subset corresponding to $i\in V_R$, and recall that $|R_j| \leq (3/4) |R_i|$ for each  $(i,j) \in E_R$, so that the depth of $R$ is $O(\log n)$. 
Moreover, as $G$ is $k$ edge-connected, each edge $(i,j) \in E_R$ has capacity $c_R(i,j)\geq k$.

Next, we convert $R$ into a {\em binary} hierarchical decomposition tree $T = (V_T,E_T)$ of $G$. The nodes of $T$ will define the sets
whose average load we will balance in the allocation strategy. To do this we will define certain flows on $T$, and we call $T$ a balancing-flow tree. 
We do not define the edge capacities $c_T$, as 
we will not view $T$ as a cut-based decomposition.
The reader may wish to assume that $R$ is already binary, and skip the paragraph below on creating the tree $T$.

\subsubsection*{Creating the tree $T$}
Apply the following  step repeatedly to $R$ until all non-leaf nodes have degree $2$.

Choose any node $i$ with $p>2$ children $j_1,\ldots,j_p$.
Let $w(h) = |S_{j_h}|/|S_i|$ for $h\in [p]$, so that $\sum_{h=1}^p w(h)=1$.

(i) If $w(h) > 1/4$ for some $h$, make $j_h$ the left child of $i$. Create a new right child $b$ with the remaining $p-1$ children other than $j_h$. The node $b$ corresponds to the set $S_i \setminus S_{j_h}$.

(ii) Otherwise, arbitrarily partition $\{j_1,\ldots,j_h\}$ into two sets $A$ and $B$ so that $w(A),w(B) \in [1/4,3/4]$. Create two children $a$ and $b$ of $i$,
where the children of $a$ (resp.~$b$) are the nodes in $A$ (resp.~$B$). The nodes $a$ and $b$ corresponds to sets $\cup_{j \in A} S_j$ and $\cup_{j \in B} S_j$.

Clearly, $T$ is binary. Also $T$ and $R$  have the same set of leaves as $R$ is a contraction of $T$. Also note that for any node $a$, if $|S_a|/|S_{p(a)}|>3/4$ where $p(a)$ is the parent of $a$ (due to step (i) above), then $a$ must be a node in the original tree R\"acke tree $R$.

The following observation will be quite useful later.
\begin{claim}\label{cl:depth}
\label{lem:ratio}
Let $a, b$ be two nodes in $T$ such that $a$ is an ancestor $b$, and let $q$ be the distance from $a$ to $b$. Then $|S_b| \leq (3/4)^{\lfloor q/2 \rfloor } |S_a|$. 
In particular, $T$ has depth $O(\log n)$. 
\end{claim}
\begin{proof}
Consider some path $c,d,e$ of length $2$ in $T$, where $c$ is the ancestor of $e$. It suffices to show that either $|S_e|\leq (3/4) |S_d|$ or $|S_d|\leq (3/4) |S_c|$.
To see this observe that if $d$ was a node in $R$, this follows by the well-balancedness of a R\"{a}cke tree for the edge $(d,e)$. Otherwise, $(c,d)$ was produced according to rule (ii) and thus $|S_e|\le |S_d|\le (3/4)|S_c|$.
\end{proof}

\subsubsection*{Defining the demands on $T$} 
Let $I_T$ denote the set of internal (non-leaf) vertices of $T$.
For $i\in I_T$, let $\ell(i)$ and $r(i)$ denote its left and right children.  
As described in Section \ref{overview}, each node $i \in I_T$ will be associated with a flow $f_i$ which will be used to balance  $\overline{L}(S_{\ell(i)})$ and  $\overline{L}(S_{r(i)})$, the average load
on subsets corresponding to its left and right children.

To this end, we define a collection of demands $D_i$ for each non-leaf node $i\in I_T$. The flow produced by these demands will determine the allocation strategy as we will describe later below.

Let $\beta = 1/(8 \alpha_G)$.
For each $i \in I_T$, the collection $D_i$ consists of $|S_{\ell(i)}| |S_{r(i)}|$ commodities
$(i,u,v)$ for each pair of vertices $u \in S_{\ell(i)}, v \in S_{r(i)}$.
For each such pair, we create the demand  from $u$ to $v$ of 
 \[\kappa_i(u,v) =  \frac{\beta k }{|S_{\ell(i)}||S_{r(i)}|}.\]
For a pair of nodes $j \in V_T$ and $i\in I_T$,  let 
\begin{equation}
\label{eq:dij}
    d_i(j) = \sum_{u \in S_j , v \notin S_j} \kappa_i(u,v) - \sum_{u \notin S_i, v\in S_j} \kappa_i(u,v), 
\end{equation}
denote the total (signed) demand from $D_i$ leaving the set $S_j$. 

Let $T_i$ denote the subtree of $T$ rooted at $i$. Note that $d_i(j) =0$ for any node $j \notin T_i$,  as the demands in $D_i$ are between the vertices in $S_i$. Also, $d_i(i)= \beta k - \beta k =0$.

These demands satisfy a useful orthogonality property, as illustrated in Figure~\ref{fig:edge flow}.
\begin{lemma}
\label{lem:orthogonal}
 For any $i,j \in I_T$, with $i\neq j$ we have that 
 \[  \frac{d_{i}(\ell(j))}{|S_{\ell(j)}|} -     \frac{d_{i}(r(j))}{|S_{r(j)}|}=0.\]
 For $j=i$ we have
  \[  \frac{d_{i}(\ell(j))}{|S_{\ell(j)}|} -     \frac{d_{i}(r(j))}{|S_{r(j)}|}= \beta k \left(\frac{1}{|S_{\ell(i)}|} + \frac{1}{|S_{r(i)}|}  \right) .\]
\end{lemma}
 \begin{proof}
We first assume that $i=j$. The property clearly holds for any node $j \notin T_i$ as $d_{i}(\ell(j)) =d_{i}(r(j)) =0$.  

Consider $j \in T_i$. As $j\neq i$, suppose without loss of generality that $j$ lies in the left subtree at $i$. Then the total demand leaving $\ell(j)$ is exactly 
 \begin{equation}
\label{eq:demij}
     d_i(\ell(j)) =  \sum_{u \in S_{\ell(j)}} \sum_{v \in S_{r(i)}} \kappa_i(u,v) =  |S_{\ell(j)}| |S_{r(i)}| \frac{ \beta k} {|S_{\ell(i)}||S_{r(i)}| } = \beta k \frac{ |S_{\ell(j)}|} {|S_{\ell(i)}|}.
 \end{equation}
 Similarly, $d_i(r(j)) =  \beta k|S_{r(j)}| /(|S_{\ell(i)}|)$  and the claim follows.  If $j$ lies in the right subtree at $i$ (so the demands enter $j$), the only difference is that sign of $d_i(\ell(j))$ is swapped.

 For $j=i$, we have $d_i(\ell(j)) = \beta_k$ and $d_i(r(j))=-\beta$ and the claim follows.
 \end{proof}

The following fact will be useful later in realizing the $T$ flow on $G$.
\begin{lemma}
\label{lem:tot-demand}
The total demand, either entering or leaving, across any edge $(a,b) \in E_T$, with $b$ child of $a$,  is at most $8 \beta k \leq  k/\alpha_G$.
\end{lemma}
\begin{proof}
The only demand leaving (resp.~entering) $b$ is due to demands in $D_i$ for nodes $i\in V_T$ that are ancestors of $b$ and for which $b$ lies in the left (resp.~right) subtree of $i$.
Let $d(b,i)$ be the distance between $b$ and $i$. Then by Claim \eqref{lem:ratio} and \eqref{eq:demij}, the total flow leaving $b$ and entering $b$ due to the demands  $D_i$ is at most 
$\beta k |S_b|/|S_{\ell(i)}| \leq \beta k (3/4)^{\lfloor (d(b,i)-1)/2 \rfloor} $.  Summing up over all the ancestors of $i$ such that $d(b,i)=1,2,\ldots$, this is at most $8 \beta k $.
\end{proof}

Now consider the total load on the edges of $R$ due to  the demands $\kappa_i(u,v)$. 
As $R$ is a contraction of $T$, the flow on an edge $(i,j) \in E_R$, where $j$ is a child of $i$, is exactly the total flow entering and leaving the node $j$ in $T$, which by the Lemma above is $8\beta k$.
In particular, this implies the following on the tree $R$.
\begin{lemma}
\label{lem:R-load}
Consider the multi-commodity flow in $R$ due to commodities $\kappa_i(u,v)$ for all $i\in I_T$ and $u,v \in V$. Then the total flow on any edge is at $8 \beta k$. As each edge in $R$ has capacity at least $k$, the congestion of any edge is at most $8\beta k/k = 1/\alpha_G$.
\end{lemma}
 \subsubsection*{Computing the balancing flows for each node of $I_T$}
 We now map these demands and the associated flow back to $G$. 
 Recall the flow templates $f_{uv}$ for each pair of vertices $u,v \in V$  in $G$, given by the oblivious routing associated with the tree $R$.
 
For each edge $e = (x,y)$ of $G$ and node $i\in I_T$, let
\[g_{i}(x,y) =  \sum_{u,v \in V} f_{uv}(x,y) \kappa_i(u,v),\]
be the total (signed) flow from $x$ to $y$ due to the demands in the collection $D_i$.
  Observe that $g_{i}(x,y) = - g_{i}(y,x)$ for all $x,y$  as $f_{uv}(x,y)=-f_{uv}(y,x)$ for all $u,v,x,y \in V$.

Note that the $g_i(x,y)$ only depend on the graph $G$ (and $R, T, \{f_{uv}\}_{uv}$ that are derived from $G$). In particular, they do not depend on the current load vector $L^t$ at the vertices of $G$, and will be fixed henceforth.

Having computed the $g_i(x,y)$, we are now ready to define the  allocation strategy.

\subsection{Allocation Strategy}
\label{sec:alloc}
At each time step $t+1$, for $t=0,1,2,\ldots$, we allocate the arriving ball using the following strategy.

\begin{enumerate} 
\item Let $e=(x,y)$ be the edge where the ball arrives at time $t+1$.
Choose $i \in I_T$ with probability 
\[p_e(i) := |g_{i}(x,y)|.\]
\item Let $L=L^{t}$ denote the current load vector.\\ Compare $\overline{L}(S_{\ell(i)})$ and  $\overline{L}(S_{r(i)})$ (the average load on  $S_{\ell(i)}$ and $S_{r(i)}$).
\begin{enumerate}
    \item If $\overline{L}(S_{\ell(i)}) > \overline{L}(S_{r(i)})$: allocate the ball to $y$ if $g_{i}(x,y)>0$, else allocate it to $x$.
\item If $\overline{L}(S_{\ell(i)}) \leq  \overline{L}(S_{r(i)})$: allocate the ball to $x$ if $g_{i}(x,y)>0$, else allocate it to $y$.
\end{enumerate}
\item With remaining probability $1-\sum_{i \in I_T} p_{e}(i) $, allocate the ball uniformly and randomly to either $x$ or $y$.
\end{enumerate}

\paragraph{Implementation.} The allocation strategy can be easily implemented to run in $O(\log n)$ time per arriving request as follows.

First, for each $e=(x,y)$, we create a table $G_e$ of size $O(n)$, with entries for each internal node $i \in I_T$, that contains $g_{i}(x,y)$. This table is fixed and does not change over time.

We maintain another table $F$ of size $O(n)$ with entries $\overline{L}(S_j)$ for each node $j\in V_T$ that change over time. Whenever a ball is allocated to some vertex $u$, we increase $\overline{L}(S_j)$ by $1/|S_j|$ for 
each node $j$ on the path from $u$ to the root $r$ in $T$. This requires only $O(\log n)$ update operations.

Upon the request for edge $e$, we lookup the entry for $i$ with probability $p_e(i)$ in the table $G_e$, and the entries $\overline{L}(S_{\ell(i)})$ and $\overline{L}(S_{r(i)})$ in the table $F$ and use the allocation strategy above to allocate the ball.

\section{Analysis}\label{sec:anal}
Our goal is to prove the bound on the gap as given by Theorem \ref{thm:main}.
To this end, we first bound the gap between the average load on sets corresponding to any two siblings in $T$. In particular, we show the following.
\begin{lemma}
\label{lem:main}
For any non-leaf node $i \in I_T$, for every constant $c>0$,
\[ \Pr [\overline{L}(S_{\ell(i)}) -  \overline{L}(S_{r(i)})  > 8( d/k) \alpha_G \, c \log n ] = O(n^{-c}).\]
\end{lemma}
Let us first see how this directly implies Theorem \ref{thm:main}.
\begin{proof}[Proof of Theorem \ref{thm:main}]
Fix any non-leaf node $i \in I_T$ and consider some child $j \in  \{\ell(i), r(i)\}$ of $i$. Let $j'$ be the sibling of $j$.
We first observe that the difference between average load of siblings is no less than the difference for a parent-child pair. This follows as
\begin{equation}
    \label{eq:rela}
    |\overline{L}(S_j) - \overline{L}(S_{j'})|
= \left|\frac{L(S_j)}{|S_j|} - \frac{L(S_i)-L(S_j)}{|S_{j'}|}\right|
= \frac{|S_i|}{|S_{j'}|} | \overline{L}(S_j) -  \overline{L}(S_i)| \geq  | \overline{L}(S_j) -  \overline{L}(S_i)|,
\end{equation} 
where we use the fact that $L(S_i) = L(S_j)+L(S_{j'})$, $|S_i|=|S_j| +|S_{j'}|$ and $|S_i| \geq |S_{j'}|$.

 Fix some vertex $u$ of $G$, and consider the path $u=i_0,i_1,\ldots,i_h=r$  in $T$ from  the leaf $u$ to the root $r$.
 As $L(u) = \overline{L}(u)$ for a leaf $u$ and $\overline{L}(r)$ is the average load over $V$,
 the deviation of load at $u$ from the average is  
 \[ |\overline{L}(u) - \overline{L}(r)| 
 = \left|\sum_{g=1}^h \big(\overline{L}(i_{g-1}) - \overline{L}(i_{g})\big)\right| \leq\sum_{g=1}^h \left|\overline{L}(i_{g-1}) - \overline{L}(i_{g})\right| .\] 
Applying Lemma \ref{lem:main} with $c>1$, taking  
a union bound over the $O(n)$ edges $(i,j) \in E_T$ and using \eqref{eq:rela}, we get that w.h.p. $|\overline{L}(S_i) - \overline{L}(S_j)|\le ( d/k) \alpha_G \, c' \log n ]$ for all edges $(i,j) \in E_T$.  As the height of $T$ is $h_T = O(\log n)$ by Claim~\ref{cl:depth}, 
this gives that,  w.h.p., for every vertex $u$,
\[|\overline{L}(u) -  \overline{L}(r)|  = O(( d/k) \alpha_G \, \log^2 n).\qedhere\]
\end{proof}

To prove Lemma \ref{lem:main} we will show that whenever $\ovl(S_{\ell(i)}) > \ovl(S_{r(i)})$, the strategy creates sufficient drift to decrease $\ovl(S_{\ell(i)}) - \ovl(S_{r(i)})$, and vice versa.

\subsection{Computing the Drifts and Biases} We now compute these quantities.

Consider the ball arriving at time $t+1$ and let $L=L^t$ denote the current load vector.
To describe the allocation strategy more compactly, it will be convenient to define 
\[Q_L(i)= \begin{cases}
 +1    &  \text{ if } \overline{L}(\ell(i)) \leq  \overline{L}(r(i)) \\
 -1  & \text{ otherwise.} 
\end{cases}
\]

For an edge $e=(x,y)$, let $q_e(x)$ denote the probability that the allocation strategy assigns the ball to $x$.
Note that $q_e(x)$ (and most of other quantities below) depend on the vector $Q_L$, but we drop this for ease of notation, as we only consider the allocation at  time $t+1$.

\begin{lemma}
$q_e(x) = \frac{1}{2} + \frac{1}{2} \sum_{i} g_{i}(x,y)) Q_L(i)$.  
\end{lemma}
\begin{proof}
Recall the allocation strategy in Section \ref{sec:alloc} can be written as follows. Upon the request $e=(x,y)$, it samples $i \in I_T$ with probability $p_e(i) = |g_i(x,y)|$ and allocates the ball to $x$ with 
probability $ (1 + \sgn(g_{i}(x,y)) Q_L(i)/2$.  Here $\sgn(a)=1$ if $a>0$ and $-1$ otherwise. So 
\begin{align*}
q_e(x)  & =   \sum_{i \in I_T} |g_{i}(x,y)| \left( \frac{1 + \sgn(g_{i}(x,y)) Q_L(i)}{2} \right) +\bigg(1 -  \sum_{i \in i_T} |g_{i}(x,y)| \bigg) \frac{1}{2} &  \\
 & =   \frac{1}{2} + \frac{1}{2} \bigg(\sum_{i\in I_T} g_{i}(x,y) Q_L(i)  \bigg). & \qedhere
\end{align*}
\end{proof}
Let us define the bias towards $x$ due to the edge $e=(x,y)$ as
\[ b_e(x) :=  2q_e(x)-1 = \sum_{i \in I_T} g_{i}(x,y)) Q_L(i).\]
We next show that the  probabilities $q_e(x,y)$ lie in $[0,1]$.
\begin{lemma}
\label{lem:well-def}
For an edge $e=(x,y)$, we have that $b_e(x) \in [-1,1]$.
\end{lemma}
\begin{proof}
As $Q_L(i)$ is $\pm 1$, 
\[|b_e(x)| \leq  \sum_{i \in I_T} |g_{i}(x,y)| \leq \sum_{i\in I_T}  \sum_{u,v} |f_{uv}(x,y)| |\kappa_i(u,v)|. \]
This is at most the total multi-commodity flow on the edge $(x,y)$
due to the demands in $D_i$ for all $i\in I_T$.
By Lemma \ref{lem:R-load}, and the property of R\"{a}cke decomposition, this is at most $\alpha_G \conge(R,\vec{d}) \leq 1$.
\end{proof}

\paragraph{The biases for sets.}
Let us now compute the probability $q(S_i)$ that the ball at time $t+1$ is assigned to a vertex in the set  $S_i$ corresponding to  some node $i\in V_T$.

Let $N(x)$ be the set of neighbors of vertex $x$ in $G$.
As $q_x(e)$ is the probability of allocating the ball to $x$ when edge $e=(x,y)$ is chosen, and each $e$ arrives uniformly with probability $1/m$,
\begin{eqnarray}
\label{qsi}
q(S_i)  = \frac{1}{m}  \sum_{x \in S_i} \sum_{y \in N(x)} \left(\frac12 + b_x(x,y)\right) 
   =  \frac{|S_i|}{n}  + \frac{1}{m} \sum_{x \in S_i} \sum_{y\in N(x)}  b_x(x,y),
\end{eqnarray}
where we use that $\sum_{x \in S_i} \sum_{y\in N(x)} 1 = d|S_i|$ as $G$ is $d$-regular, and that $m=nd/2$.

Let us define
\begin{equation}
\label{bsi}
    b(S_i):=\frac{1}{m} \sum_{x \in S_i} \sum_{y\in N(x)}  b_x(x,y)=q(S_i)-\frac{|S_i|}{n}
\end{equation} as the bias towards set $S_i$, over the stationary probability $|S_i|/n$.

The following key lemma that relates the demands $D_j$ defined on the nodes $j\in I_T$ to the bias $b(S_i)$.
Recall from \eqref{eq:dij}, that  
$d_j(i)$ is the total flow out of the subtree $T_i$ due to the demands in $D_j$.

\begin{lemma}
\label{lem:set-bias}
For any set $S_i$ corresponding to a node $i$ of $T$
\[ b(S_i)  = \frac{1}{m} \sum_{j\in i_T} Q_L(j) d_j(i). \]
\end{lemma}
\begin{proof}
By the definition of $b_x(x,y)$ and $g_i(x,y)$,
\begin{eqnarray*}  \sum_{x \in S_i} \sum_{y \in  N(X)}  b_x(x,y) & = &  \sum_{x \in S_i} \sum_{y \in N(x)} \sum_{j \in i_T} g_{j}(x,y)) Q_L(j) \\   
 &  = &   \sum_{x \in S_i} \sum_{y \in N(x)} \sum_{j\in i_T} Q_L(j) \sum_{u,v} f_{uv}(x,y) d_j(u,v)  \\
 &  = &  \sum_{j\in I_T} Q_L(j)  \sum_{u,v} d_j(u,v)   \sum_{x \in S_i} \sum_{y \in N(x)} f_{uv}(x,y).   \end{eqnarray*}
Noting that $\sum_{y\in N(x)} f_{uv}(x,y)  = 1_{x=u}-1_{x=v}$ for any $x$,
\[ \sum_{x \in S_i} \sum_{y \in N(x)} f_{uv}(x,y) = \sum_{x \in S_i} (1_{x=u}-1_{x=v}) = 1_{u\in S_i} - 1_{v \in S_i}.\]
Plugging back this gives that
\[  \sum_{x \in S_i} \sum_{y \in N(x)}  b_x(x,y) = \sum_{j\in i_T} Q_L(j)  \sum_{u,v} d_j(u,v) (    1_{u\in S_i} - 1_{v \in S_i}) = \sum_{j\in i_T} Q_L(j) d_j(i). \qedhere \]
\end{proof}

The following quantity determines the relative bias between two siblings and will be crucial in the proof of Lemma \ref{lem:main}.
\begin{lemma}
\label{qsi-normalized-diff}
For a node $i$ and its children $\ell(i)$ and $r(i)$,
\[ \frac{q(S_{\ell(i)})}{|S_{\ell(i)}|} - \frac{q(S_{r(i)})}{|S_{r(i)}|} = \frac {\beta k Q_L(i)}{m} \left(\frac{1}{|S_{\ell(i)}|}+   \frac{1}{|S_{r(i)}|}\right).
\]
\end{lemma}
\begin{proof}
By definition for bias in \eqref{bsi}, we have  $q(S_{\ell(i)}) = |S_{\ell(i)}|/n + b(S_{\ell(i)})$ and 
$q(S_{r(i)}) = |S_{r(i)}|/n + b(S_{r(i)})$.
By Lemma \ref{lem:set-bias}.
\[
\frac{q(S_{\ell(i)})}{|S_{\ell(i)}|} - \frac{q(S_{r(i)})}{|S_{r(i)}|}  = 
\frac{b(S_{\ell(i)})}{|S_{\ell(i)}|} - \frac{b(S_{r(i)})}{|S_{r(i)}|} = 
\sum_{j \in I_T} \frac{1}{m} Q_L(j) \left( \frac{ d_j(\ell(i))}{|S_{\ell(i)}|} - \frac{ d_j(r(i))}{|S_{r(i)}|} \right). \]
By Lemma \ref{lem:orthogonal}, for any $i \neq j$ the right side is zero, the only contribution is due to $j=i$ which is
 \[  \frac{\beta k Q_L(i)}{m} \left(\frac{1}{|S_{\ell(i)}|} +  \frac{1}{|S_{r(i)}|} \right). \qedhere  \]
\end{proof}
\subsection{Proving the concentration}
We  now prove Lemma \ref{lem:main}.
We first develop a concentration lemma to be used in our analysis. 
\paragraph{2-point concentration.} The set up is the following. There are two bins, $1$ and $2$, associated with so called
steady state probabilities $\pi_1, \pi_2$ which satisfy $\pi_1+\pi_2=1$. 
At each time step a ball arrives and
let $\ell^t_1, \ell^t_2$ denote the loads of the bins at the end of time $t$. In addition, fix $b\leq \min(\pi_1,\pi_2)/2$, and assume that the allocation process at $t+1$ is as follows:
a parameter $b_t>b$ also satisfying $b_t\le \max(\pi_1,\pi_2)/2$ is chosen independent of the future random choices.  Then, if $\ell^t_1/\pi_1 > \ell^t_2/\pi_2$ the ball is allocated to either bin 1 with probability $\pi_1-b_t$ or to bin 2 with probability $\pi_2+b_t$. Otherwise, if $\ell^t_1/\pi_1 \leq \ell^t_2/\pi_2$, the ball is allocated to bin 1 with probability $\pi_1 + b_t$, or bin 2 with probability $\pi_2-b_t$.

For any time $t$, let
$d(t) = \ell^t_1/\pi_1 - \ell^t_2/\pi_2$ denote the normalized difference in the loads of the bin. 
\begin{lemma}
\label{lem:2pt}
For any time $t$ and any $x\geq 0$ it holds that 
\[\Pr[|d(t)| \geq  x/b] = O\big(\exp(-x/8)\big).\]
\end{lemma}
\begin{proof}
We define the potential function  
$\Phi(t) =  \cosh (\alpha d(t))$, where  $\alpha = b/8$.
So $\Phi(0) =1$ initially at $t=0$.  

Fix a time $t$ and let $\Delta \Phi = \Phi(t+1) - \Phi(t)$ and denote $d=d(t)$ and  $\Delta d = d(t+1)  - d(t)$.
Then, by Taylor expansion, and using $(d/dx) \cosh(x) = \sinh(x)$ and $(d/dx) \sinh x = \cosh x$,  
\[\Delta \Phi = \sinh (\alpha d) \Big(\alpha \Delta d + (\alpha \Delta d)^3/3! + \ldots\Big)  + \cosh (\alpha d)  \Big(\alpha (\Delta d)^2/2! + (\alpha \Delta d)^4/4! + \ldots\Big). \]
Let $\gamma =1/\pi_1  + 1/\pi_2$. As $b \leq \max(\pi_1,\pi_2)/2$ we have that $\gamma b \leq 1$. As either $\ell_1^t$ or $\ell_2^t$ rises by $1$, we have $|\Delta d| \leq \gamma$, and hence $|\alpha \Delta d| \leq 1/2$.
Bounding $|(\alpha \Delta d)^i | \leq (\alpha \Delta d)^2 2^{-i+2}$ for $i\geq 2$, and using that $ \cosh x -1 \leq |\sinh(x)| \leq \cosh x$ for all $x$,
\[  \Delta \Phi \leq \sinh (\alpha d) (\alpha \Delta d) + \cosh(\alpha d) (\alpha \Delta d)^2 .\] 
For $d>0$,
\[\E (\Delta d)  =  (\pi_1 - b_t)\frac{1}{\pi_1} - (\pi_2 + b_t)\frac{1}{\pi_2}   =  -b_t \left(\frac{1}{\pi_1} + \frac{1}{\pi_2}\right) = -b_t\gamma \le -b \gamma\]
and otherwise,  for $d\leq 0$, we have $\E[\Delta d] = b_t \gamma\ge b \gamma$.

For $d>0$, we have
\[\E[(\Delta d)^2] =  (\pi_1 - b_t) \left(\frac{1}{\pi_1}\right)^2 + (\pi_2+ b)\left(\frac{1}{\pi_2}\right)^2 = \gamma  + b_t(-1/\pi_1^2 + 1/\pi_2^2) \leq \gamma + b_t \gamma^2 \leq 2 \gamma,\]
where the last step uses the fact that $b_t \gamma \leq 1$. The same bound also holds for $d\leq 0$.

This gives
\begin{eqnarray*} \E[\Delta \Phi(t)\ |\ \Phi(t)] & \leq &  \alpha \sinh (\alpha d) \E[\Delta d] + 2 \alpha^2 \cosh(\alpha d) \E[(\Delta d)^2] \\
& \leq  &  - \alpha |\sinh (\alpha d) | b \gamma + 4 \alpha^2 \gamma \cosh(\alpha d) \\
& \leq & -\alpha b \gamma (\Phi(t)-1) + \alpha b \gamma \Phi(t)/2 = -\alpha b \gamma (\Phi(t) -2) /2.
\end{eqnarray*}
Thus, taking expectation and using $\Delta(\Phi(t))=\Phi(t+1)-\Phi(t)$ we obtain
\[\E[\Phi(t+1)]\le \alpha b \gamma +(1-\alpha b \gamma/2)\E[\Phi(t)].\]
Taking induction over $t$ and recalling that $\Phi(0)=1$ we obtain 
\[\E[\Phi(t)]\le 2\big(1-(1-\tfrac{\alpha b\gamma}{2})^t\big)+(1-\tfrac{\alpha b\gamma}{2})^t\le 2.\]
The result now follows by applying Markov's inequality to $\Phi(t)$.
\end{proof}

\subsubsection*{Proof of Lemma \ref{lem:main}.}
We apply the set up above to bound $|\ovl(S_{\ell(i)})-\ovl(S_{r(i)})|$. Let bin 1 correspond to $S_{\ell(i)}$ and bin 2 to $S_{r(i)}$. We only focus on the time steps when the ball is allocated to some vertex in  $S_i$. Let $\pi_1 = |S_{\ell(i)}|/|S|$ and $\pi_2 = |S_{r(i)}|/|S|$ so that $\pi_1+\pi_2=1$.
The probability of allocating the ball to $S_{\ell(i)}$ conditioned on it being allocated to $S_i$ is $q(S_{\ell(i)})/q(S_i)$, and similarly $q(S_{r(i)})/q(S_i)$ for $S_{r(i)}$. Recall that $q$ depends on the load vector (and hence the time) but we suppress these for ease of notation.

Writing $\pi_1 -b_t = q(S_{\ell(i)})/q(S_i)$ and $\pi_2 + b_t = q(S_{r(i)})/q(S_i)$,  gives that 
\[ b_t \left(\frac1{\pi_1}+\frac1{\pi_2}\right)  = \frac{ q(S_{r(i)})}{ \pi_2 q(S_i)}  - \frac{ q(S_{\ell(i)})}{\pi_1 q(S_i)} = \frac{|S_i|}{q(S_i)} \left( \frac{ q(S_{r(i)})}{ |S_{r(i)}|} -   \frac {q(S_{\ell(i)})}{ |S_{\ell(i)}|}    \right),  \]
which by Lemma \ref{qsi-normalized-diff}  equals
\[- \frac{|S_i|}{q(S_i)} \frac{\beta k Q_L(i)}{m} \left(\frac{1}{|S_{\ell(i)}|} +  \frac{1}{|S_{r(i)}|} \right) = - \frac{1}{q(S_i)} \frac{\beta k Q_L(i)}{m} \left(\frac{1}{\pi_1} +  \frac{1}{\pi_2} \right),\]
which gives that 
\begin{equation}
\label{eq:bt}
    b_t = - \frac{\beta k Q_L(i)}{m q(S_i)} 
\end{equation}
We claim that $|b_t| \geq \beta k Q_L(i)/m q(S_i)$. This follows as  $Q_L(i) \in \{-1,1\}$ and $q(S_i) = |S_i|/n + b(S_i) \leq 2 |S_i|/n$, where we use that by Lemma \ref{lem:set-bias} and Lemma \ref{lem:tot-demand}, $b(S_i) \leq 8 \beta k/m = 2k/(\alpha_G nd) \leq 2/(\alpha_G n)  \leq |S_i|/n$.

Let us define 
\begin{equation}\label{eq:bdef}
b :=\beta k/(d|S_i|) |b_t| = \frac{\beta k n}{2m |S_i|},\end{equation}
so that $b \leq |b_t|$.

Applying Lemma \ref{lem:2pt}, we obtain
\[ \Pr\left[ \left| \frac{L(S_{\ell(i)})}{\pi_1} - \frac{L(S_{r(i)})}{\pi_2}  \right| \geq x/b \right] = O\big(\exp(-x/8)\big).\]
As $\pi_1 = |S_{\ell(i)}|/|S_i|,  \pi_2 = |S_{\ell(i)}|/|S_i|$  and using the definition of $b$ in \eqref{eq:bdef}, gives that for any $x>0$,
\[  \Pr\left[ | \ovl(S_{\ell(i)}) - \ovl(S_{r(i)})  | \geq  8x(d/k) \alpha_G\right] = O\big(\exp(-x/8)\big),\]
which implies the desired result.

\section{Lower Bounds}
\label{sec:lb}
We describe various simple but instructive lower bounds. First we prove Theorem \ref{thm:lb} in Section \ref{sec:lb1}. Then, in Section \ref{sec:lb2} we show an $\Omega(\log(n)/\log \log n)$ lower bound on the upper gap for bounded degree graphs. In particular this implies that on bounded degree expanders, the upper gap is roughly of the same order as  the gap, in contrast to the complete graph where the upper gap is $O(\log \log n)$ and the gap is $\Omega(\log n)$. 

We also remark that \cite{PTW15} showed that for complete graphs, under the $1+\beta$ choice model for a fixed $\beta<1$, the upper gap is $\Omega(\log n)$ and hence similar to the gap. Again this is in contrast to $\beta=1$ (the 2-choice model) where the upper gap  $O(\log \log n)$ and gap is $\Omega(\log n)$. 

\subsection{Proof of Theorem \ref{thm:lb}}
\label{sec:lb1}
We show for any $d$-regular $k$ edge-connected graph $G$, under any allocation strategy, the gap is $\Omega(d/k +\log n)$  with at least constant probability.

The $\Omega(\log n)$ bound follows from the following folklore argument. Fix any time $t$ and 
consider any interval $I$ of $O(n \log n)$ steps just before $t$. For a fixed vertex, at any time step, the probability that some edge incident to it is chosen is $d/m = 2/n$. So by standard coupon-collector argument, with $\Omega(1)$ probability, there is some vertex $v$ for which no incident edge is chosen during $I$. So during $I$, the load of $v$ cannot change under any strategy, while the average load increases by $\log n$. 

Hence, to prove Theorem \ref{thm:lb}, it suffices to show the following.
\begin{lemma}
Let $G$ be any $d$-regular, $k$ edge-connected graph. Then for any allocation strategy, and at any time $t$ the gap is $\Omega(d/k)$ with constant probability.
\end{lemma} 
\begin{proof}
Let $C=(S,\overline{S})$ be some minimum cut in $G$, and let $S$ be the larger side with $|S|\geq n/2$.
Fix a time $t$ and 
consider some time interval $I$ of length $T$ just before $t$. We specify $T$ later. Let $Y$ be the number of edges requested with both endpoints in $S$.
As an edge is requested uniformly at random at each time, and as there are $(d|S|-k)/2$ such edges, we obtain
\[ \E[Y] = (d|S|-k)/2 \cdot T/m =  (|S|-k/d)T/n.\]
As $k \leq d$ and  $|S|\geq n/2$, this at least $T/3$, provided that $n\geq 6$.

By standard estimates, with at least constant positive probability, we have $Y \geq  \E[Y] + \sqrt{T}$.
Conditioned on this event, the average load on vertices in $S$ during $I$ exceeds the stationary load of $T/n$ by at least
\[  \frac{Y}{|S|} - \frac{T}{n} \geq  \frac{\E[Y] + \sqrt{T}}{|S|} - \frac{T}{n} =   \frac{(|S|d-k)T}{|S|dn} + \frac{\sqrt{T}}{|S|} -  \frac{T}{n} = \frac{\sqrt{T}}{|S|}  - \frac{k T}{nd|S|}.\]
Choosing $T = cn^2d^2/k^2$ for small enough $c$ and as $|S|\geq n/2$, this is
$\Omega( d/k)$.
\end{proof}
\subsection{Upper gap for bounded degree graphs}
\label{sec:lb2}
Next we show that even for bounded degree expanders, under the greedy strategy the maximum load at time $t$ is typically $t/n + \tilde{\Omega}(\log n)$. This is unlike for complete graphs where the upper gap is $O(\log \log n)$.
\begin{lemma}
For any $d$-regular graph $G=(V,E)$ with $d=O(1)$, at any time $t$ and for any allocation strategy, with constant probability the maximum vertex load is $t/n + \Omega(\log n/\log \log n)$. 
\end{lemma}
\begin{proof} Fix a time $t$, and consider the interval $I$ of length $|E|=dn/2$ just before $t$. Let $t'=t-|I|$ be time at the beginning of $I$, and let $a = t'/n$ denote the average load at $t'$. Let $s=\log n/(4\log\log n)$.
We will show that with constant probability, either $t$ or $t'$ has upper gap $\tilde{\Omega}(\log n)$.

If some vertex has load $a+s/4$ at time $t'$, then we are already done. Otherwise, by an averaging argument, there can be at most $n/4$ vertices with load at most $ a-s$ at time $t'$. Call an edge {\em bad} if it is incident to such a vertex, and let $B$ be the set of bad edges.
As $G$ is $d$-regular, $|B| \leq d(n/4) \leq |E|/2$.

As $|I|=|E|$, and each edge is requested uniformly and independently during $I$, with constant probability, some edge $e\in E\setminus B$ will be requested at least $4s$ times. If this event occurs, then for any allocation strategy, the load on some endpoint of $e$ increases by at least $2s$ during $I$, and hence becomes at least $a-s+2s=a+s$ at time $t$. On the other hand, the average load only increases by $|I|/n=d/2 = O(1)$, to become $a+d/2$ at time $t$. This implies the result.
\end{proof}

\subsection*{Concluding Remarks}
Our result in this paper gives, in particular, an allocation strategy with asymptotically polylogarithmic gap for any $d$-regular $d$-connected graph. One may wonder whether this strategy gives any improvement in practical settings where the number of bins is not too large. To show that it indeed does, we conclude the paper with with a simulation result of our strategy for cycles, compared against the asymptotic behavior of the greedy strategy (Figure~\ref{fig:Flow Gap}).

\begin{figure}[ht!]
    \centering
    \includegraphics[scale=0.5]{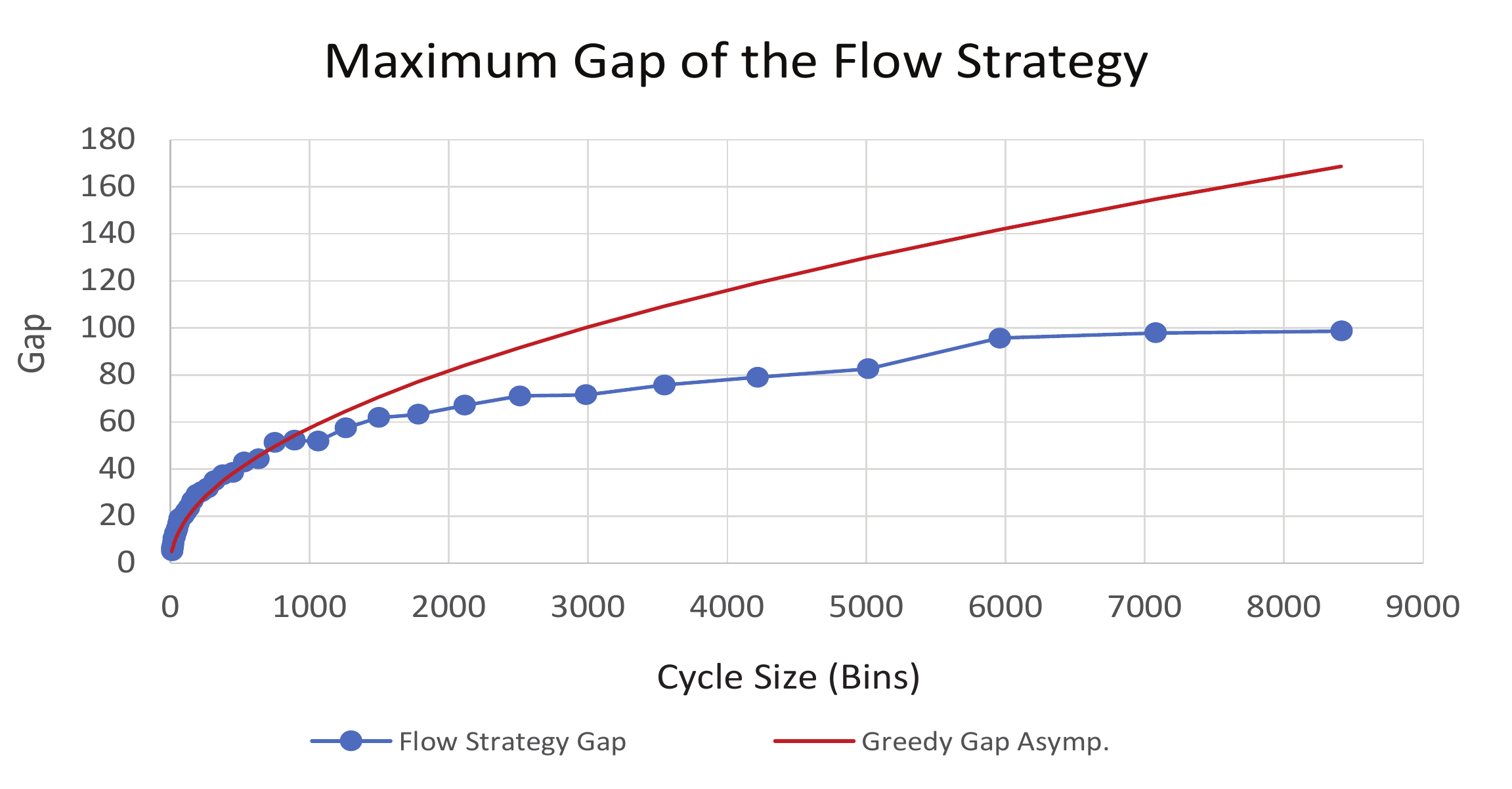}
    \caption{The gap for our flow algorithm on cycles for the best possible R\"{a}cke tree, averaged over 32 runs of $n^{2.5}$ balls for cycles of sizes $n=10$ to $8400$.   
    }
    \label{fig:Flow Gap}
\end{figure}

\section*{Acknowledgments}
We thank  Guy Bensky and Shlomo Ron for their coding assistance in making our simulations.
\small{
\bibliographystyle{alpha}
\bibliography{ref}
}
\end{document}